\documentclass[
journal
]{IEEEtran}
\usepackage{amsmath,amssymb,amsfonts}
\usepackage{amsthm}
\usepackage{algorithm}
\usepackage[noend]{algorithmic}
\usepackage{graphicx}
\usepackage{textcomp}
\usepackage{caption2}
\usepackage{balance}
\def\BibTeX{{\rm B\kern-.05em{\sc i\kern-.025em b}\kern-.08em
    T\kern-.1667em\lower.7ex\hbox{E}\kern-.125emX}}
\usepackage{hyperref}
\usepackage{subfigure}
\usepackage{tikz} 
  \usetikzlibrary{shapes.geometric, arrows}
  \usetikzlibrary{arrows} 
  \usetikzlibrary{decorations.pathmorphing} 
  \usetikzlibrary{backgrounds} 
  \usetikzlibrary{positioning} 
  \usetikzlibrary{fit} 
  \usetikzlibrary{petri}

\newtheorem{definition}{\bf Definition}
\newtheorem{example}{\bf Example}
\newtheorem{theorem}{\bf Theorem}
\newtheorem{lemma}{\bf Lemma}
\newtheorem{remark}{\bf Remark}
\newtheorem{proposition}{\bf Proposition}

\newtheorem{corollary}{\bf Corollary}
\DeclareMathOperator{\N}{\mathbb{N}}
\DeclareMathOperator{\supN}{\rm supN}
\DeclareMathOperator{\supCN}{\rm supCN}
\DeclareMathOperator{\supCNP}{\rm supCNP}

\DeclareMathOperator{\supCND}{\rm supCND}
\DeclareMathOperator{\supCNN_p}{\rm supCNN_p}
\DeclareMathOperator{\supP}{\rm supP}
\DeclareMathOperator{\supNN_p}{\rm supNN_p}

\newcommand*{\stopnumbering}{%
  \let\olditem\item
  \renewcommand{\item}[1][]{\olditem[]}%
  \let\State\Statex}
\newcommand*{\resumenumbering}{%
  \let\item\olditem
  \let\State\oldState}

\def\highlight#1{\textcolor{blue}{#1}}
\begin{document}

\title{Active Prognosis and Diagnosis of Modular Discrete-Event Systems}

\author{Shaopeng Hu\textsuperscript{1}, Shaowen Miao\textsuperscript{2}, Jan Komenda\textsuperscript{3}, and Zhiwu Li\textsuperscript{1,4}
\thanks{This research is supported by the National R\&D Program of China under Grant No. 2018YFB1700104, by the Science Technology Development Fund, MSAR under Grant No. 0029/2023/RIA1, and by the Czech Academy of Sciences under RVO 67985840. \textit{(Corresponding author: Zhiwu Li.)}}
\thanks{\textsuperscript{1} Shaopeng Hu is with the School of Electro-Mechanical Engineering, Xidian University, Xi'an 710071, China (e-mail: \textit{sphu@stu.xidian.edu.cn}).}
\thanks{\textsuperscript{2} Shaowen Miao is with Robotics and Autonomous Systems Thrust, The Hong Kong University of Science and Technology (Guangzhou), 511453, China (e-mail: \textit{smiao585@connect.hkust-gz.edu.cn}).}
\thanks{\textsuperscript{3}Jan Komenda is with the Institute of Mathematics of the Czech Academy of Sciences, 115 67 Prague, Czechia (e-mail: \textit{komenda@ipm.cz}).}
\thanks{\textsuperscript{4}Zhiwu Li is with the School of Electro-Mechanical Engineering, Xidian University, Xi'an 710071, China, and also with the Institute of Systems Engineering, Macau University of Science and Technology, Taipa, Macau (e-mail: \textit{zwli@must.edu.mo}).} 
}

\maketitle

\begin{abstract}
    This paper addresses the verification and enforcement of prognosability and diagnosability for discrete-event systems (DESs) modeled by deterministic finite automata. We establish the equivalence between prognosability (respectively, diagnosability) and pre-normality over a subset of the non-faulty language (respectively, a suffix of the faulty language). We then demonstrate the existence of supremal prognosable (respectively, diagnosable) and normal sublanguages. Furthermore, an algorithm is then designed to compute the supremal controllable, normal, and prognosable (respectively, diagnosable) sublanguages.
    Since DESs are typically composed of multiple components operating in parallel, pure local supervisors are generally insufficient, as prognosability and diagnosability are global properties of a system. Given the limited work on enforcing prognosability or diagnosability in modular DESs, where these properties are enforced through local supervisors, this paper leverages a refined version of pre-normality to compute modular supervisors for local subsystems. The resulting closed-loop system is shown to be globally controllable, normal, and prognosable/diagnosable. Examples are provided to illustrate the proposed method. 
\end{abstract}


\section{Introduction}
    Fault diagnosis and diagnosability enforcement of discrete-event systems (DESs) have attracted attention from both researchers and practitioners over the past decades. These methods have found applications in domains such as manufacturing systems~\cite{sampath1995diagnosability}, transportation and communication networks~\cite{CEPpaper}, and smart grids~\cite{smartgrid}. In such systems, faults are not directly observable due to their inherent nature or limitations in sensor deployment. If not identified and addressed within a reasonable timeframe, undetected faults can result in severe consequences.

    The objective of \emph{fault diagnosis} is to determine whether a fault has occurred by analyzing the observable outputs of a plant~\cite{sampath1995diagnosability}. A system is \emph{diagnosable} if every fault can be detected after a finite number of observable events following its occurrence. When a plant is not diagnosable, it must be modified before deployment---a process known as \emph{diagnosability enforcement}. When enforcement is achieved through supervisory control that actively alters system behavior to facilitate fault detection, the process is referred to as \emph{active diagnosis}~\cite{Sampath1998activedia}.

    Fault diagnosis and diagnosability verification have been approached using integer linear programming for Petri nets~\cite{Dotoli2009ilp} and automata-theoretic methods for finite automata~\cite{sampath1995diagnosability, Lin1994diagnosability, Jiang2001verifier, Yoo2002verifier, miao2025decentralized}. The \emph{diagnoser} was introduced to derive necessary and sufficient conditions for diagnosability~\cite{sampath1995diagnosability, Lin1994diagnosability}, but its construction incurs exponential complexity in the plant's state space. To address this, the \emph{verifier} was proposed as a polynomial-time alternative~\cite{Jiang2001verifier, Yoo2002verifier}.

    An extension of diagnosis, known as \emph{fault prognosis} or \emph{predictability}, aims to anticipate faults before they occur~\cite{Gencpro12006, Basilepro22009}. A system is \emph{prognosable} if every fault can be predicted ahead of time based on observations. A quantified variant, called $k$-prognosis, seeks to determine whether a fault will occur within the next $k$ observable steps~\cite{cassez2013predictability, Chouchane2024kpro}, with larger values of $k$ generally preferred.

    Prognosability verification has been extensively studied. Genc and Lafortune~\cite{Genc09pro} propose a polynomial-time verifier-based approach for centralized systems. J\'erom et al.~\cite{jer08pro} introduce diagnoser-based methods for verifying fault patterns in finite transition systems. Extensions to stochastic automata and Petri nets have also been explored~\cite{sdes3pro, LefebvrestoPNpro}, as well as the approaches tailored to timed systems, where the remaining time to fault occurrence replaces the number of observable steps~\cite{cassez2013predictability}.
    Recent developments have expanded prognosability to unbounded Petri nets~\cite{yin2018unboundedpro}, decentralized frameworks~\cite{decpro1, decpro2}, and distributed settings~\cite{dispro}. 

    Supervisory control offers a formal framework for enforcing specifications by disabling controllable events based on observable behavior. Active diagnosis in deadlock-free systems has been investigated~\cite{Paoli2005supervisor, yin2015uniform} and extended to systems with deadlocks~\cite{hu2025iot,Hu2020supervisor, Hu2023dt}. Furthermore, active prognosis is introduced and proven to be EXPTIME-complete via a game-theoretic approach in \cite{haar2020active}. 

    An alternative enforcement strategy involves modifying the observation structure. For example, Markov decision theory has been applied to select cost-effective observations for diagnosability~\cite{Debouk2002sensor}. In stochastic DESs, sensors can be dynamically enabled or disabled~\cite{Thorsley2007sensor}. Another approach uses event relabeling, where a relabeling function for diagnosability enforcement is designed by solving an integer linear programming~\cite{Cabasino2013sensor, Ran2019sensor, Hu2024sensor}. Diagnosability enforcement via event relabeling has also been extended to various models, including timed systems~\cite{velasquez2022active, miao2025active}, unbounded Petri nets~\cite{Hu2024unb}, and attack-prone DESs~\cite{hu2025smc}.

    Unlike monolithic systems, modular DESs consist of interacting components operating concurrently. Diagnosability in such systems, termed \emph{modular diagnosability}, has been studied in several works~\cite{Schmidt2013modulardia, Masopust2019modulardia, Basilio2023modulardia, Contant2006modulardia}. A key challenge lies in synthesizing local supervisors such that the overall behavior of the controlled modules mimics that of a nonblocking, maximally permissive centralized supervisor. In modular DESs, diagnosability verification and enforcement become more complex. Local supervisors must not only control behavior but also detect faults in a way that ensures diagnosability at the system level, without deviating from the globally desired behavior.

    
    In this paper, we characterize prognosability/diagnosability in terms of extensions of the faulty language. To this end, we establish a connection between prognosability (resp. diagnosability) and (pre-)normality, showing how both properties can be interpreted within a unified framework based on the observability of marked behaviors in partially observable systems. Based on these two connections, we design a maximally permissive supervisor to enforce prognosability/diagnosability via (pre-)normality by computing the supremal normal and pre-normal sublanguage.

    Although prognosability and diagnosability can be verified in polynomial time in the number of states of the plant using verifier~\cite{Yoo2002verifier, Jiang2001verifier, Ranproenforcement2022}, the exponential growth of state space in modular systems renders such methods impractical. Furthermore, active diagnosis and prognosis are EXPTIME-complete~\cite{bertrand2014active}, and hence no polynomial-time algorithms are expected. While modular diagnosability verification has been studied~\cite{Schmidt2013modulardia, Masopust2019modulardia, Basilio2023modulardia, Contant2006modulardia}, no prior work has addressed modular active diagnosis or prognosis.
    
    In this paper, we focus on the verification and enforcement of prognosability (resp. diagnosability) in (modular) DESs.
    The main contributions of this paper are as follows:
    \begin{enumerate}
        \item We discuss the properties of prognosability and diagnosability, and provide a novel characterization of prognosability and diagnosability in terms of normality or pre-normality.

        \item We prove the existence of the supremal prognosable/diagnosable and normal sublanguage, and develop an algorithm to compute the supremal controllable, normal, and prognosable/diagnosable sublanguage in monolithic DESs.

        \item This paper provides sufficient conditions to enforce modular prognosability/diagnosability directly from local models, thus avoiding the explicit construction of the global plant.
    \end{enumerate}

    Section~\ref{Sec: Preliminary} reviews preliminary concepts from automata and supervisory control theory. Section~\ref{Sec: Diagnosability} introduces prognosability and diagnosability and establishes their connection to pre-normality. The existence of supremal prognosable/diagnosable and normal sublanguages is shown in Section~\ref{Sec: Existence}. Section~\ref{Sec: Active diagnosis} presents the supervisor synthesis method for the supremal controllable, normal, and prognosable/diagnosable sublanguages. Section~\ref{Section:modularenforcement} extends these results to modular DESs. Section~\ref{Sec: Conclusion} concludes the paper and outlines the directions for future work.

\section{Preliminaries and Concepts}\label{Sec: Preliminary}
    In this section, we overview definitions and results from (modular) supervisory control of deterministic finite automata~\cite{cassandras2021introduction,wonham2019supervisory,komenda2023modular} and discuss the properties of diagnosability and ($k$-step) prognosability. 

\subsection{Strings, languages, and automata}
    The cardinality of a set $A$ is denoted by $|A|$. An alphabet, $\Sigma$, is a finite nonempty set of events. The set of finite strings over $\Sigma$ is denoted by $\Sigma^*$, including the empty string denoted by $\varepsilon$. The length of a string $s\in \Sigma^*$ is denoted by $|s|$. The set of prefixes of $s\in \Sigma^*$ is denoted by $\overline{s}=\{ s'\in \Sigma^* \mid \exists t\in \Sigma^*: s=s't\}$.

    A language is a subset of $\Sigma^*$. The set of prefixes of a language $L$ is denoted by $\overline{L} = \bigcup_{s\in L} \overline{s}$. A language $L$ is \emph{prefix-closed} if $L = \overline{L}$. A language $K \subseteq L$ is \emph{extension-closed} w.r.t. $L$ if $K \Sigma^* \cap L \subseteq K$. Given a sublanguage $K\subseteq L$, it holds that $K$ is extension-closed, i.e., $K \Sigma^* \cap L \subseteq K$ if and only if $L\setminus K$ is prefix-closed, i.e., $L\setminus K=\overline{L\setminus K}$.

    The \emph{left quotient} of a language $L$ w.r.t. a language $L'$ is defined as $L'\backslash L = \{t\in\Sigma^* \mid \exists s\in L': st\in L\}$. Analogously, the \emph{right quotient} of $L$ w.r.t. $L'$ is $L/ L' = \{t\in\Sigma^* \mid \exists s\in L':ts\in L\}$.

    A \textit{projection} $P: \Sigma^* \to \Sigma_o^*$ for $\Sigma_o \subseteq \Sigma$ is a morphism defined by $P(\sigma) = \sigma$ for $\sigma \in \Sigma_o$ and $P(\sigma) = \varepsilon$ for $\sigma \in \Sigma \setminus \Sigma_o$. It removes the events that are not in the event set $\Sigma_o$. The inverse projection $P^{-1}: \Sigma_o \to 2^{\Sigma^*}$ is defined by $P^{-1}(t) = \{s \in \Sigma^* \mid P(s) = t \}$. These definitions can be readily extended to languages.
    We denote $\Sigma_o^{\leq N}=\{t\in \Sigma_o^*\mid |P(t)|\leq N\}$ and $\Sigma_o^{\geq N}=\{t\in \Sigma_o^*\mid |P(t)|\geq N\}$.

    A \emph{deterministic finite automaton} (DFA) is a quintuple $G = (Q, \Sigma, \delta, q_0, Q_m)$, where $Q$ is a finite set of states, $\Sigma$ is an alphabet, $\delta: Q \times \Sigma \to Q$ is a transition function that can be extended to $\delta: Q \times \Sigma^* \to Q$ in a usual way, $q_0 \in Q$ is the initial state, and $Q_m\subseteq Q$ is the set of marked states. We write $G = (Q, \Sigma, \delta, q_0)$ if the set of marked states $Q_m$ is irrelevant. The \textit{generated} and \textit{marked languages} of $G$ are $L(G) = \{ s \in \Sigma^* \mid \delta(q_0, s) \in Q \}$ and $L_m(G) = \{ s \in \Sigma^* \mid \delta(q_0, s) \in Q_m \}$, respectively. We have  $L(G) = \overline{L(G)}$ and $L_m(G)\subseteq L(G)$. 
    
    A DFA $G$ is \emph{nonblocking} if $L(G)=\overline{L_m(G)}$, it is \emph{live} if for every state $q\in Q$, there is an event $\sigma\in\Sigma$ such that $\delta(q,\sigma)$ is defined, and it is  {\em convergent} if it does not contain cycles of unobservable events. In what follows, the term language refers to a language marked by a DFA.


    The observer of $G = (Q, \Sigma, \delta, q_0, Q_m)$, denoted by $Obs(G)$, is the accessible part of the DFA obtained by the standard subset construction from the automaton created from $G$ by replacing every unobservable event with $\varepsilon$~\cite{cassandras2021introduction}. 
    

\subsection{Basic concepts of supervisory control}
    Given a DFA $G$ over $\Sigma$, the alphabet $\Sigma$ is partitioned into \emph{observable events} $\Sigma_o$ and \emph{unobservable events} $\Sigma_{uo}$, and into \textit{controllable events} $\Sigma_c$ and \textit{uncontrollable events} $\Sigma_{uc}$. 
    The set of \emph{control patterns} is defined by $\Gamma = \{ \gamma \subseteq \Sigma \mid \Sigma_{uc} \subseteq \gamma \}$. The \textit{supervisor} of $G$ is a map $S: P(L(G)) \to \Gamma$. The behavior of the \textit{closed-loop system}, denoted by $L(S/G)$, is defined by $\varepsilon \in L(S/G)$, and iff $s \in L(S/G)$, $s \sigma \in L(G)$, and $\sigma \in S(P(s))$, then $s \sigma \in L(S/G)$. Intuitively, observing $P(s)$, the supervisor disables the controllable events from $\Sigma_c \setminus S(P(s))$. 

    Usually, it is impossible to attain any given language as the behavior of a closed-loop system. However, controllable and observable languages can be realized~\cite{WonhamLin1988}.
    A language $M \subseteq L(G)$ is \textit{controllable} w.r.t. $L(G)$ and $\Sigma_{uc}$ if $\overline{M}\Sigma_{uc} \cap L(G) \subseteq \overline{M}$, and it is \textit{observable} w.r.t. $L(G)$, $P: \Sigma^* \to \Sigma^*_o$, and $\Sigma_c$ if, for every $s \in \overline{M}$ and every $\sigma \in \Sigma_c$, $s \sigma \notin \overline{M}$ and $s \sigma \in L(G)$ imply $P^{-1}[P(s)]\sigma \cap \overline{M} = \emptyset$. 
    Since observability, unlike controllability, is not preserved under language unions, a stronger notion of normality was introduced \cite{wonham2019supervisory}. A language $M\subseteq L(G)$ is \textit{normal} w.r.t. $L(G)$ and $P$ if $\overline{M} = P^{-1}[P(\overline{M})]\cap L(G)$. Normality is a property of prefix-closure of a language, while 
    in \cite{wonham2019supervisory}, pre-normality is introduced as a property of the language itself. Recall that  $M\subseteq L(G)$ is \textit{pre-normal} w.r.t. $L(G)$ and $P$ if $M = P^{-1}[P(M)]\cap L(G)$. Note that pre-normality of $M\subseteq L(G)$ is equivalent to normality of $M$ if $M$ is prefix-closed. 
    
    The supremal normal sublanguage of $M$ w.r.t. $L(G)$ and $P$, denoted by $\supN(M,L(G),P)$, is equal to the union of all sublanguages of $M$ that are normal w.r.t. $L(G)$ and $P$. Similarly, we denote by $\supCN(M,L(G),\Sigma_{uc}, P)$ the supremal controllable and normal sublanguage of $M$ w.r.t. $L(G)$, $\Sigma_{uc}$, and $P$.
    Moreover, pre-normality has symmetry, i.e.,
$M\subseteq L(G)$ is pre-normal w.r.t. $L(G)$ and $P$ if and only if $L(G)\setminus M $ is pre-normal w.r.t. $L(G)$ and $P$.

\subsection{Modular supervisory control}
    Most systems are modeled as a synchronous product of several subsystems. 
    The synchronous product of languages $L_i$ over  $\Sigma_i$ is the language $\|^l_{i=1} L_i = \bigcap^l_{i=1} P^{-1}_i(L_i)$ over $\Sigma = \bigcup^l_{i=1} \Sigma_i$, where $P_i: \Sigma^* \to \Sigma^*_i$ is the projection to local alphabet $\Sigma_i$ for $i=1,\dots,l$. 
    For DFAs $G_i$ over $\Sigma_i$, there is a DFA $G = \ \parallel^l_{i=1} \! G_i$ over $\Sigma$ satisfying $L(\parallel^l_{i=1} \! G_i) = \ \parallel^l_{i=1} \! L(G_i)$. 
    The languages $L_i$ are \textit{nonconflicting} if $\overline{\parallel^l_{i=1} \! L_i} = \ \parallel^l_{i=1} \! \overline{L_i}$.

    Given DFAs $G_i$, $i=1,2, \dots,l$, generating languages $L_i = L(G_i)$ with the global behavior $L = \|^l_{i=1} L_i$, and a specification $M \subseteq L$,
    the objective of the modular control problem is to synthesize local supervisors $S_i$ such that $\parallel^l_{i=1} \! L(S_i / G_i) = L(S/ \! \parallel^l_{i=1} \! G_i)$, where $S$ is a supervisor of the specification $M$ and the global plant language $L$~\cite{komenda2023modular}.

    Let $P_i: \Sigma^* \to \Sigma^*_i$ denote the projections to modules, and let the corresponding observations and projections of modules be $P^i_{i,o}: \Sigma^*_i \to (\Sigma_i \cap \Sigma_o)^*$ and $P^o_{i,o}: \Sigma^*_o \to (\Sigma_i \cap \Sigma_o)^*$, see Fig.~\ref{Fig: Notation}.
    The local observable events are denoted by $\Sigma_{i,o} = \Sigma_i \cap \Sigma_o$. 
    We assume that the events observable in one component are observable in all components where they appear, i.e., $\Sigma_{i,o} \cap \Sigma_j = \Sigma_i \cap \Sigma_{j,o} = \Sigma_{i,o} \cap \Sigma_{j,o}$. For a modular system $G = \ \parallel^l_{i=1} \! G_i$, with $G_i$ over $\Sigma_i$, the set of shared events is defined by $\Sigma_s = \bigcup_{i \not= j} (\Sigma_i \cap \Sigma_j).$

    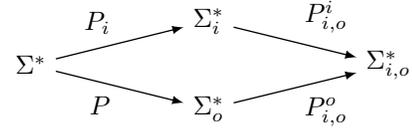
\begin{figure}
        \centering
        \begin{tikzpicture}[%
    auto,
    >=stealth',
    node distance=2.5cm
]

\node (A) {$\Sigma^*$};
\node (B) [above right=0.05cm and 1.7cm of A] {$\Sigma_i^*$};
\node (C) [below right=0.05cm and 1.7cm of A] {$\Sigma_o^*$};
\node (D) [right=4cm of A] {$\Sigma_{i,o}^*$};

\draw[-latex,line width=0.5pt] (A) -- node[above left] {$P_i$} (B);
\draw[-latex,line width=0.5pt] (A) -- node[below left] {$P$} (C);
\draw[-latex,line width=0.5pt] (B) -- node[above right] {$P^i_{i,o}$} (D);
\draw[-latex,line width=0.5pt] (C) -- node[below right] {$P^o_{i,o}$} (D);

\end{tikzpicture}
        \caption{Projection notations in this paper.}\label{Fig: Notation}
    \end{figure}


    

\subsection{Diagnosability, prognosability and $k$-prognosability}
    Given a DFA $G$ over $\Sigma$, we use $\Sigma_f\subseteq (\Sigma\setminus \Sigma_o)$ to denote the set of \emph{fault events}. 
    Moreover, the fault events are divided into several different classes. For the sake of simplicity, we consider only one class of faults, which is not restrictive~\cite{sampath1995diagnosability}.

    The \emph{faulty language of $G$} is $L_f=\Sigma^* \Sigma_f \Sigma^*\cap L(G)$ and the non-faulty language is its complement $L_n=L(G) \setminus L_f$. Let $\Psi (\Sigma_{f}) = L\cap \Sigma^* \Sigma_f$ denote the set of strings that end with a fault event.
    We adopt the common assumption used in the literature on diagnosability and prognosability.

    \begin{quote}
    {\bf (A1)} The considered DFAs are live and convergent.
    \end{quote}
    
    We now recall the definitions of diagnosability \cite{sampath1995diagnosability}, prognosability \cite{Chouchane2024kpro}, and $k$-prognosability \cite{cassez2013predictability}. 
    \begin{definition}[Diagnosability]\label{Def: dia}
        A live and convergent DFA $G$ is \emph{diagnosable} w.r.t. projection $P$ and the set of fault events $\Sigma_f$ if there exists a natural number $n$ such that, for every fault-ending string $s\in \Psi(\Sigma_f)$ and every extension $t\in s\backslash L(G)$ of length at least $n$, every string $w\in P^{-1}P(st)\cap L(G)$ contains a fault. \hfill$\diamond$
    \end{definition}

    \begin{definition}[Prognosability]\label{Def: pro}
        A live and convergent DFA $G$ is \textit{prognosable}  w.r.t. $P$ and $\Sigma_f$ if every $s\in \Psi(\Sigma_f)$ has a prefix $s'\in \overline{s}$ such that, for every $t\in P^{-1}P(s')$ that contains no fault, there is a natural number $n$ for which every extension $t'\in t\backslash L(G)$ of length at least $n$ contains a fault.
        \hfill$\diamond$
    \end{definition}

    \begin{definition}[$k$-prognosability]\label{Def: Kpro}
        Given a natural number $k$, a live and convergent DFA $G$ is \textit{$k$-prognosable} w.r.t. $P$ and $\Sigma_f$ if every fault-ending string $s\in \Psi(\Sigma_f)$ has a prefix $s'\in \overline{s}$ such that $|P(s)|-|P(s')| = k$ and, for every $t\in P^{-1}P(s')$ that contains no fault, there is a natural number $n$ for which every extension $t'\in t\backslash L(G)$ of length at least $n$ contains a fault.
        \hfill$\diamond$
    \end{definition}

    Intuitively, $k$-prognosability requires that every fault is predicted at $k$ steps before it occurs. From the security perspective, a larger value of $k$ is preferred to ensure more reliable and proactive fault prediction. It is known that prognosability implies diagnosability~\cite{Genc09pro}. In the following, we discuss the relationship between Definitions \ref{Def: pro} and \ref{Def: Kpro}. 
    
    \begin{lemma} \label{Lemma:k-1pro}
        If $G$ is a live and convergent DFA, $s\in \Psi(\Sigma_f)$ is a fault-ending string, and $s'$ is a prefix of $s$ that satisfies prognosability, then every $s's''$ that is a prefix of $s$ also satisfies prognosability.
    \end{lemma}
    \begin{proof}
        If $r\in P^{-1}P(s's'')$ is a non-faulty string, then $r$ has a non-faulty prefix $r'\in P^{-1}P(s')$. By the definition of prognosability, there is a natural number $n\in \N$ such that every extension of $r'$ within $L(G)$ of length at least $n$ contains a fault. However, every extension of $r$ within $L(G)$ of length at least $n$ is also an extension of $r'$ of length at least $n$. Thus, $s's''$ satisfies prognosability.
    \end{proof}

    It is worth noting that, different from the concept of $k$-step (observation) prognosability, the study in \cite{cassez2013predictability} introduces $K$ time units prognosability within the framework of timed automata. In \cite{cassez2013predictability}, conclusions similar to Lemma~\ref{Lemma:k-1pro} can be found, although they are not directly comparable to ours. Lemma~\ref{Lemma:k-1pro} has the following two consequences.
    \begin{corollary} \label{Corollary:k-1pro}
        For every natural number $k>0$, if a live and convergent DFA $G$ is $k$-prognosable w.r.t. $P$ and $\Sigma_f$, then it is $(k-1)$-prognosable w.r.t. $P$ and $\Sigma_f$. 
    \end{corollary}
    \begin{proof}
        It directly follows from Lemma~\ref{Lemma:k-1pro}. 
    \end{proof}
    
By Corollary~\ref{Corollary:k-1pro}, $k$-prognosability implies $0$-prognosability. In the following, we establish that $0$-prognosability is in fact equivalent to the conventional notion of prognosability.

    \begin{corollary} \label{Corollary:equiv}
        Let $G$ be a live and convergent DFA. Then, $G$ is prognosable w.r.t. $P$ and $\Sigma_f$ if{f} $G$ is $0$-prognosable w.r.t. $P$ and $\Sigma_f$. 
    \end{corollary}
    \begin{proof}
        If $G$ is $k$-prognosable w.r.t. $P$ and $\Sigma_f$, then it is prognosable w.r.t. $P$ and $\Sigma_f$ by definition. On the other hand, if $G$ is prognosable w.r.t. $P$ and $\Sigma_f$, then it is $0$-prognosable w.r.t. $P$ and $\Sigma_f$ by Lemma~\ref{Lemma:k-1pro}.
    \end{proof}

    
    The following lemma characterizes the negation of $k$-prognosability. From the definition of negation of $k$-prognosability, to test whether a plant $G$ is non-prognosable, one needs to consider a fault-ending string $s\in \Psi(\Sigma_f)$ and all its prefixes such that the condition of Definition~\ref{Def: Kpro} does not hold.
    However, by the property of prefix, actually, we only need to consider a prefix $s'\in \overline{s}$.

    \begin{lemma} \label{Lem:non-k-pro}
        Given a natural number $k$, a live and convergent DFA $G$ is not $k$-prognosable w.r.t. $P$ and $\Sigma_f$ if and only if there is a fault-ending string $ss' \in \Psi(\Sigma_f)$ with $|P(s')|= k$ and a non-faulty string $t\in P^{-1}P(s)$, such that for every natural number $n$, there is a non-faulty extension $t'\in t\backslash L(G)$ of length at least $n$.  
    \end{lemma}

\begin{proof}
    (If) We prove that $G$ is not $k$-prognosable by showing that the string $ss' \in \Psi(\Sigma_f)$ from the statement of the lemma violates $k$-prognosability; namely, we show that for every prefix $w$ of $ss'$ with $|P(ss')|-|P(w)| = k$, there is a non-faulty string $r\in P^{-1}P(w)$ such that, for every $n\in\mathbb{N}$, there is a non-faulty extension $r'\in r\backslash L(G)$ of length $n$. To this end, let $w$ be a prefix of $ss'$. 
    If $P(w)=P(s)$, then the string $t\in P^{-1}P(s) = P^{-1}P(w)$ from the condition of the lemma completes the proof.
    If $P(w)\neq P(s)$, the claim holds vacuously as  $|P(ss')|-|P(w)| \neq k$.

    (Only if) Assume that $G$ is not $k$-prognosable. By Definition~\ref{Def: Kpro}, there is a fault-ending string $w\in \Psi(\Sigma_f)$ such that, for every prefix $w'\in \overline{w}$ with $|P(w)|-|P(w')|= k$, there is a non-faulty string $t\in P^{-1}P(w')$ and for all $n\in \N$, there is $t'\in t\backslash L(G)$ with $|t'|\geq n$ and $tt'\notin L_f$. This clearly implies the condition from the statement of the lemma by choosing $s=w'$ for some of those $w'\in \overline{w}$ with $|P(w)|-|P(w')|= k$ and $ss'=w$. 
\end{proof}

\section{Characterization of Diagnosability and Prognosability}\label{Sec: Diagnosability}

In this section, we show that prognosability and diagnosability can be characterized in terms of pre-normality. This characterization will further be employed in the subsequent sections to verify the existence of, and to compute, the supremal normal and $k$-prognosable/diagnosable sublanguages. 

\subsection{Characterizations of prognosability}

To characterize $k$-prognosability in terms of pre-normality, let 
\begin{multline*}
    \Psi_f^{-k}=\{ut\in \overline{\Psi(\Sigma_f)} \mid  \exists s's\in \Psi(\Sigma_f): |P(s)| \leq k~ \land\\ 
    u\in P^{-1}P(s')\}\,.
\end{multline*}
This corresponds to the completion within the prefix-closure of the language leading to the first fault (by string $t$) of all strings $u$ that look like strings (here $s'$) that are less than $k$ observations from the occurrence of the first fault,
which can be equivalently expressed using the right quotient operation as
   \begin{equation}
        \Psi_f^{-k}=P^{-1}P[ \Psi(\Sigma_f) / P^{-1}(\Sigma_o^{\leq k})]\Sigma^* \cap \overline{\Psi(\Sigma_f)}, \label{(1111)}
          \end{equation}
where $\Sigma_o^{\leq N}=\{t\in \Sigma_o^*\mid |P(t)|\leq N\}$. Consequently, there is a DFA marking the language $\Psi_f^{-k}$.
In the following, we show the property of language $\Psi_f^{-k}$, which is used to simplify the proof of the following proposition.

\begin{lemma} \label{Lemma:extension-closed}
Given a natural number $k\in \N$, a DFA $G$, and the set of fault-ending strings $\Psi(\Sigma_f)$, then $\Psi_f^{-k}$ is pre-normal w.r.t. $\overline{\Psi(\Sigma_f)}$ and $P$, and is extension-closed w.r.t. $\overline{\Psi(\Sigma_f)}$, as well as $\overline{\Psi(\Sigma_f)}\setminus \Psi_f^{-k}$ is prefix-closed. 
\end{lemma}

\begin{proof}
We first show that $\Psi_f^{-k}$ is pre-normal w.r.t. $\overline{\Psi(\Sigma_f)}$ and $P$. We need to show $P^{-1}P(\Psi_f^{-k})\cap \overline{\Psi(\Sigma_f)}=\Psi_f^{-k}$ by the definition of pre-normality. It holds that 
\begin{align*}
&P^{-1}P(\Psi_f^{-k})\cap \overline{\Psi(\Sigma_f)}=\\
&P^{-1}P\big[P^{-1}P[ \Psi(\Sigma_f) / P^{-1}(\Sigma_o^{\leq k})]\Sigma^* \cap \overline{\Psi(\Sigma_f)} \big] \cap \overline{\Psi(\Sigma_f)}=\\
&P^{-1}P[ \Psi(\Sigma_f) / P^{-1}(\Sigma_o^{\leq k})]\Sigma^* \cap P^{-1}P(\overline{\Psi(\Sigma_f)}) \cap \overline{\Psi(\Sigma_f)}=\\
&P^{-1}P[ \Psi(\Sigma_f) / P^{-1}(\Sigma_o^{\leq k})]\Sigma^* \cap \overline{\Psi(\Sigma_f)}=\Psi_f^{-k}.
\end{align*}

    We then show that $\Psi_f^{-k}$ is extension-closed w.r.t. $\overline{\Psi(\Sigma_f)}$. For the sake of brevity, let $K= P^{-1}P[ \Psi(\Sigma_f) / P^{-1}(\Sigma_o^{\leq k})]$. We have $\Psi_f^{-k}=K\Sigma^* \cap \overline{\Psi(\Sigma_f)}$. Then it holds $[K\Sigma^* \cap \overline{\Psi(\Sigma_f)}]\Sigma^*\cap \overline{\Psi(\Sigma_f)}=K\Sigma^* \cap \overline{\Psi(\Sigma_f)}\Sigma^*\cap \overline{\Psi(\Sigma_f)}=K\Sigma^* \cap \overline{\Psi(\Sigma_f)}=\Psi_f^{-k}$, which implies that $\Psi_f^{-k}$ is extension-closed w.r.t. $\overline{\Psi(\Sigma_f)}$.

    Now we prove that $\overline{\Psi(\Sigma_f)}\setminus \Psi_f^{-k}$ is prefix-closed. By contradiction, there is a string $s\in \overline{\Psi(\Sigma_f)}\setminus \Psi_f^{-k}$ such that there exists $s'\in \overline{s}$ with $s'\notin \overline{\Psi(\Sigma_f)}\setminus \Psi_f^{-k}$. Since $\overline{\Psi(\Sigma_f)}$ is prefix-closed, we have $s'\in \Psi_f^{-k}$. According to $\Psi_f^{-k}$ is extension-closed w.r.t. $\overline{\Psi(\Sigma_f)}$, for all $s't\in \overline{\Psi(\Sigma_f)}$ we have  $s't\in \Psi_f^{-k}$. By $s'\in \overline{s}$, we have $s\in \Psi_f^{-k}$, which leads to a contradiction and completes the proof.
\end{proof}  

According to Lemmas \ref{Lem:non-k-pro} and \ref{Lemma:extension-closed}, we characterize $k$-prognosability in terms of pre-normality.

\begin{proposition}\label{Pro:prodef}
    Given a number $k\in \N$, a live and convergent DFA $G$ is $k$-prognosable w.r.t. $\Sigma_f$ and $P$ if and only if $P^{-1}P(\overline{L_n\setminus \Psi_f^{-k}})\cap L \subseteq L_n\setminus \Psi_f^{-k}$.
\end{proposition}
    \begin{proof}
    (If) We show by contrapositive that if $G$ is not $k$-prognosable, then the inclusion does not hold. By Lemma~\ref{Lem:non-k-pro}, there is a composed string $s's\in \Psi(\Sigma_f)$ with $|P(s)|= k$ and $t\in L_n$ with $P(t)=P(s')$ such that, for all $i\in \mathbb{N}$, there is a string $t_i\in t\backslash L$ with $|t_i|\geq i$ and $tt_i\notin L_f$. Then one sees $s'\in \Psi_f^{-k}$. If, for all $i\in \mathbb{N}$, $tt_i\in L_n\setminus\Psi_f^{-k}$, we have $t\in \overline{L_n\setminus \Psi_f^{-k}}$, which implies $s' \in P^{-1}P(\overline{L_n\setminus \Psi_f^{-k}})\cap L$. Since $s'\in \Psi_f^{-k}$, we have $s' \notin L_n \setminus \Psi_f^{-k}$, which shows that $P^{-1}P(\overline{L_n\setminus \Psi_f^{-k}})\cap L \not\subseteq L_n\setminus \Psi_f^{-k}$.

    Now we consider that there is an integer $i\in \mathbb{N}$ such that $tt_i\in \Psi_f^{-k}$. By $s'\in \Psi_f^{-k}$ and $P(t)=P(s')$, $t\in \Psi_f^{-k}$ holds. Furthermore, for all $j\in \N$, there exists $tt_j\in L_n$ with $|t_j|\geq j$. Since one can take $j> |tt_i|$, we have $tt_j\in L_n\setminus \Psi_f^{-k}$, i.e., $t\in \overline{L_n\setminus \Psi_f^{-k}}$ but $t\notin L_n\setminus\Psi_f^{-k}$, which again leads to 
$P^{-1}P(\overline{L_n\setminus \Psi_f^{-k}})\cap L\not \subseteq L_n\setminus \Psi_f^{-k}$.

    (Only if) First, let us consider that there exists $m\in \N$ such that for all $s \in L_n$ we have $|s|\leq m$. Since $G$ is live, we have $L_n=\overline{\Psi(\Sigma_f)}\setminus \Psi(\Sigma_f)$
     and $G$ is always $k$-prognosable. As $\overline{\Psi(\Sigma_f)}\setminus \Psi_f^{-k}$ is prefix-closed by Lemma \ref{Lemma:extension-closed}, $\overline{L_n\setminus \Psi_f^{-k}}=\overline{\overline{\Psi(\Sigma_f)}\setminus \Psi_f^{-k}}= \overline{\Psi(\Sigma_f)}\setminus \Psi_f^{-k}=L_n\setminus\Psi_f^{-k}$. By the symmetry of pre-normality, $P^{-1}P(\overline{\Psi(\Sigma_f)}\setminus\Psi_f^{-k})\cap \overline{\Psi(\Sigma_f)}=\overline{\Psi(\Sigma_f)}\setminus\Psi_f^{-k}$. According to the definition of $\Psi_f^{-k}$ and $L_n=\overline{\Psi(\Sigma_f)}\setminus \Psi(\Sigma_f)$, $L_n\setminus\Psi_f^{-k}=\overline{\Psi(\Sigma_f)}\setminus\Psi_f^{-k}$ holds.
We have $P^{-1}P(L_n\setminus\Psi_f^{-k})\cap L= P^{-1}P(L_n\setminus\Psi_f^{-k})\cap L\cap \overline{\Psi(\Sigma_f)}\subseteq L_n\setminus\Psi_f^{-k}$. 
 By viewing $L_n\setminus\Psi_f^{-k}$ is prefix-closed, it gives $P^{-1}P(\overline{L_n\setminus \Psi_f^{-k}})\cap L \subseteq L_n\setminus \Psi_f^{-k}$.

    Now, we consider the case where $L_n$ contains arbitrarily long strings. Assume that the formula does not hold, i.e., there is $t \in (P^{-1}P(\overline{L_n\setminus \Psi_f^{-k}})\cap L)  \setminus (L_n\setminus \Psi_f^{-k})$. Then, there is $s'\in \overline{L_n\setminus \Psi_f^{-k}}$ such that $P(s')=P(t)$. 
    Consider $s'\in \overline{L_n\setminus \Psi_f^{-k}}\setminus (L_n\setminus \Psi_f^{-k})$. We have $s'\in \Psi_f^{-k}$ and for all $i\in \N$, there exists $t_i\in L/s'$ with $|t_i|\geq i$ such that $s't_i\in L_n$. Since $\Psi_f^{-k}\subseteq \overline{\Psi(\Sigma_f)}$, by definition of  $\Psi_f^{-k}$, there exists $u'u\in \Psi(\Sigma_f)$ with $P(u')=P(s')$ and $|P(u)|\leq k$. Due to Lemma~\ref{Lem:non-k-pro}, $G$ is not $k$-prognosable.
    
    We assume that $s'\in L_n\setminus \Psi_f^{-k}$ and recall that $P(s')=P(t)$. Either $t\in \Psi_f^{-k}$ or $t\in L_f$ holds. For the former, by definition of $\Psi_f^{-k}$, if $u\in P^{-1}P(t)$ and there exists $v\in \Sigma^*$ with $uv\in \Psi(\Sigma_f)$, we have $u\in \Psi_f^{-k}$. By pre-normality of $\Psi_f^{-k}$ w.r.t. $\overline{\Psi(\Sigma_f)}$, for all $i\in \N$, there exists $t_i\in \Sigma^*$ with $P(t_i)\geq i$, $s't_i\in L_n$. Further, there exists $t_1t_2\in \Psi(\Sigma_f)$ with $P(t_1)=P(t)$ and $|P(t_2)|\leq k$.
    Then $G$ is not $k$-prognosable by Lemma~\ref{Lem:non-k-pro}. We consider $s'\in L_n\setminus \Psi_f^{-k}$ and $t\in L_f$. There exist $s''\in \overline{s'}$ and $t'\in \overline{t}$ such that $t'\in \Psi(\Sigma_f)$ and $P(s'')=P(t')$. This is equivalent to the case of $s'\in \overline{L_n\setminus \Psi_f^{-k}}$ and $t\in \Psi_f^{-k}$. We conclude that $G$ is $k$-prognosable w.r.t. $\Sigma_f$ and $P$ if and only if $P^{-1}P(\overline{L_n\setminus \Psi_f^{-k}})\cap L \subseteq L_n\setminus \Psi_f^{-k}$.
\end{proof}

Proposition \ref{Pro:prodef} can be reformulated as follows. Note that pre-normality of $M\subseteq L(G)$ is equivalent to normality of $M$ if $M$ is prefix-closed.

\begin{corollary} \label{Corollary: prenormal and closed}
   Given a number $k\in \N$, a live and convergent DFA $G$ is $k$-prognosable w.r.t. $\Sigma_f$ and $P$ if and only if the language $L_n\setminus \Psi_f^{-k}$ is prefix-closed, and is pre-normal (normal) w.r.t. $L(G)$ and $P$.
\end{corollary}

By Corollaries \ref{Corollary:equiv} and \ref{Corollary: prenormal and closed}, $G$ is prognosable w.r.t. $\Sigma_f$ and $P$ if and only if sublanguage $L_n\setminus \Psi_f^{-0}$ is prefix-closed, and is pre-normal (normal) w.r.t. $L(G)$ and $P$.

Proposition \ref{Pro:prodef} also holds for its complements because pre-normality is symmetric w.r.t. complements. Specifically, a plant $G$ is $k$-prognosable w.r.t. $P$ and $\Sigma_f$ if for $k\in \N$, $P^{-1}P((L_f\cup \Psi_f^{-k})\Sigma^*)\cap L \subseteq L_f\cup \Psi_f^{-k}$. Similar to Proposition~\ref{Pro:prodef}, i.e., language $L_n\setminus \Psi_f^{-k}$ is normal and prefix-closed, we require that $L_f\cup \Psi_f^{-k}$ should be pre-normal and extension-closed, i.e., all the extensions of strings are in itself ($L_f\cup \Psi_f^{-k}=(L_f\cup \Psi_f^{-k})\Sigma^* \cap L$).

\begin{example} \label{Exa:prodef}
Given a DFA $G_1=(Q_1, \Sigma_1, \delta_1, q_{0,1})$ depicted in Fig. \ref{Fig:G_1}(a), the event set is $\Sigma_1=\{a,b,c,f_1,\tau\}$ with fault event set $\Sigma_{1,f}=\{f_1\}$, where $a$, $b$, and $c$ are observable, and $\tau$ is non-faulty unobservable. The set of strings ending with a fault is $\Psi(\Sigma_{1,f})=\{abbf_1,cbbf_1\}$. Since there exist prefixes $ab\in \overline{abbf_1}$ and $cb\in \overline{cbbf_1}$ with $|P(abbf_1)|-|P(ab)|=1$ and $|P(cbbf_1)|-|P(cb)|=1$ such that $P^{-1}P(ab)=\{ab\}$ and $P^{-1}P(cb)=\{cb\}$, $G_1$ is $1$-prognosable w.r.t. $\Sigma_{1,f}$ and $P$ due to Definition~\ref{Def: Kpro}. 

Now, let us test $G_1$ by Proposition \ref{Pro:prodef}. We consider first $k=2$. By the definition of $\Psi_f^{-k}$, we have $\Psi_f^{-2}=\{a,c, ab,$ $cb, abb,cbb,abbf_1,cbbf_1\}$. According to $\overline{\Psi(\Sigma_{1,f})}=\overline{abbf_1}\cup \overline{cbbf_1}$, we have $\Psi_f^{-2}$ is pre-normal and extension-closed w.r.t. $\overline{\Psi(\Sigma_{1,f})}$ by the definitions of pre-normality and extension-closed language, respectively. 
Furthermore, we have $\overline{L_n\setminus \Psi_f^{-2}}=$ $L_n\setminus \Psi_f^{-2} =\{\varepsilon,\tau, \tau a, \tau aa, \tau aabc^j\}$ for $j\in \mathbb{N}$. It holds that $a\in P^{-1}P(\overline{L_n\setminus \Psi_f^{-2}})$ but $a\notin L_n\setminus \Psi_f^{-2}$ due to $P(a)=P(\tau a)$. By Proposition \ref{Pro:prodef}, $G_1$ is not $2$-prognosable w.r.t. $\Sigma_{1,f}$ and $P$.

Then let us consider $k=1$. We have $\Psi_f^{-1}=\{ab,cb, abb,cbb,abbf_1,cbbf_1\}$ and $\overline{L_n\setminus \Psi_f^{-1}}= L_n\setminus \Psi_f^{-1} =\{\varepsilon,\tau, a,\tau a,\tau aa,\tau aabc^j\}$ for $j\in \mathbb{N}$. According to $P^{-1}P(\overline{L_n\setminus \Psi_f^{-1}})\cap L \subseteq L_n\setminus \Psi_f^{-1}$, $G_1$ is $1$-prognosable w.r.t. $\Sigma_{1,f}$ and $P$ by Proposition \ref{Pro:prodef}. 
    \hfill$\blacksquare$
\end{example}

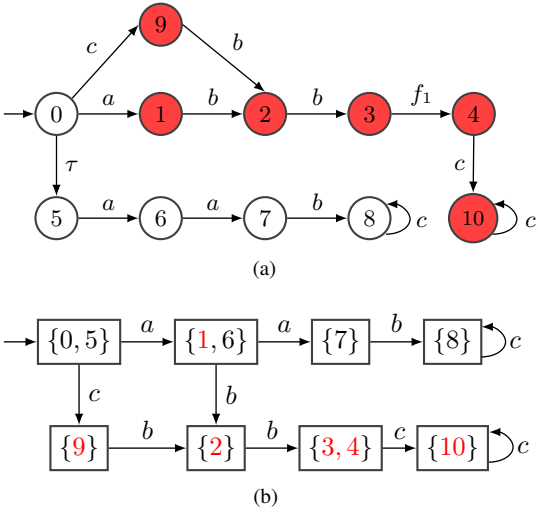
\begin{figure}[htbp]
\subfigure[]
{
    \centering
	\begin{tikzpicture}{every node}=[font= \fontsize{9pt}{10pt},scale=1]
	\node (c0) [circle,thick,draw=black!75] {$0$};
	\node (c1) [circle,thick,draw=black!75,fill=red!75,right=0.8cm of c0] {$1$};
	\node (c2) [circle,thick,draw=black!75,fill=red!75,right=0.8cm of c1] {$2$};
    \node (c3) [circle,thick,draw=black!75,fill=red!75,right=0.8cm of c2] {$3$};
	\node (c4) [circle,thick,draw=black!75,fill=red!75,right=0.8cm of c3] {$4$};
	\node (c5) [circle,thick,draw=black!75,below=0.8cm of c0] {$5$};
	\node (c6) [circle,thick,draw=black!75,right=0.8cm of c5] {$6$};
    \node (c7) [circle,thick,draw=black!75,right=0.8cm of c6] {$7$};
	\node (c8) [circle,thick,draw=black!75,right=0.8cm of c7] {$8$};
    \node (c9) [circle,thick,draw=black!75,fill=red!75,above=0.6cm of c1] {$9$};
    \node (c10) [circle,thick,draw=black!75,fill=red!75,right=0.75cm of c8,font=\footnotesize] {$10$};

        \draw[-latex,line width=0.5pt] (-0.7,0) -- (c0.west);
	\draw[-latex,line width=0.5pt] (c0.east) --node [auto]{$a$} (c1.west);
    \draw[-latex,line width=0.5pt] (c0.north east) --node [auto]{$c$} (c9.west);
    \draw[-latex,line width=0.5pt] (c9.east) --node [auto]{$b$} (c2.north);
	\draw[-latex,line width=0.5pt] (c1.east) --node [auto]{$b$} (c2.west);
    \draw[-latex,line width=0.5pt] (c2.east) --node [auto]{$b$} (c3.west);
    \draw[-latex,line width=0.5pt] (c3.east) --node [auto]{$f_1$} (c4.west);
	\draw[-latex,line width=0.5pt] (c0.south) --node [auto]{$\tau$} (c5.north);
    \draw[-latex,line width=0.5pt] (c6.east) --node [auto]{$a$} (c7.west);
    \draw[-latex,line width=0.5pt] (c7.east) --node [auto]{$b$} (c8.west);
	\draw[-latex,line width=0.5pt] (c5.east) --node [auto]{$a$} (c6.west);
    \draw[-latex,line width=0.5pt] (c4.south) --node [auto,swap]{$c$} (c10.north);
	\draw[-latex,line width=0.5pt] (4.38,-1.6)arc[start angle=-90, end angle=90, x radius=.32, y radius=.2];
	\node[black,below] at (4.85,-1.25) {{$c$}};

    \draw[-latex,line width=0.5pt] (5.8,-1.6)arc[start angle=-90, end angle=90, x radius=.32, y radius=.2];
	\node[black,below] at (6.3,-1.25) {{$c$}};


	\end{tikzpicture}}
    \subfigure[]{
    \centering
\begin{tikzpicture}
    \node (c0) [draw,thick,draw=black!75] {$\{0,5\}$};
    \node (c1) [draw,thick,draw=black!75, right=0.7cm of c0] {$\{\textcolor{red}{1},6\}$};
    \node (c2) [draw,thick,draw=black!75, right=0.7cm of c1] {$\{7\}$};
    \node (c3) [draw,thick,draw=black!75, right=0.7cm of c2] {$\{8\}$};
    \node (c5) [draw,thick,draw=black!75, below=0.8cm of c1] {$\{\textcolor{red}{2}\}$};
    \node (c6) [draw,thick,draw=black!75, right=0.7cm of c5] {$\{\textcolor{red}{3,4}\}$};
    \node (c8) [draw,thick,draw=black!75, below=0.8cm of c0] {$\{\textcolor{red}{9}\}$};
    \node (c9) [draw,thick,draw=black!75, below=0.8cm of c3] {$\{\textcolor{red}{10}\}$};

    \draw[-latex,line width=0.5pt] (-1,0) -- (c0.west);
    \draw[-latex,line width=0.5pt] (c0.east) --node [auto]{$a$} (c1.west);
    \draw[-latex,line width=0.5pt] (c1.east) --node [auto]{$a$} (c2.west);
    \draw[-latex,line width=0.5pt] (c2.east) --node [auto]{$b$} (c3.west);
    \draw[-latex,line width=0.5pt] (c1.south) --node [auto]{$b$} (c5.north);
    \draw[-latex,line width=0.5pt] (c5.east) --node [auto]{$b$} (c6.west);
    \draw[-latex,line width=0.5pt] (c8.east) --node [auto]{$b$} (c5.west);
    \draw[-latex,line width=0.5pt] (c0.south) --node [auto]{$c$} (c8.north);
    \draw[-latex,line width=0.5pt] (c6.east) --node [auto]{$c$} (c9.west);

    \draw[-latex,line width=0.5pt] (5.35,-0.2) arc[start angle=-90, end angle=90, x radius=.32, y radius=.2];
    \node[black,below] at (5.8,0.2) {{$c$}};
    \draw[-latex,line width=0.5pt] (5.45,-1.6) arc[start angle=-90, end angle=90, x radius=.32, y radius=.2];
    \node[black,below] at (5.9,-1.2) {{$c$}};

	\end{tikzpicture}
    }
    \caption{(a) A DFA $G_1$ and (b) its observer $Obs(G_1)$.}
    \label{Fig:G_1}
\end{figure}

\begin{example}
    We adapt an observer-based method to verify the $k$-prognosability of $G_1$. This method is further employed in subsequent sections to compute the supremal normal and $k$-prognosable sublanguage. Let the marked states be the states reached by firing strings in $L_f\cup \Psi_f^{-2}$, i.e., $Q_{m}=\{1,2,3,$ $4,9,10\}$, which are depicted in red in Fig.~\ref{Fig:G_1}(a). The observer of $G_1$ is shown in Fig.~\ref{Fig:G_1}(b). Since there exists an observer state that contains both marked and non-marked states, i.e., $\{1,6\}$, we conclude that language $L_f\cup \Psi_f^{-2}$ is not pre-normal w.r.t. $L(G_1)$ and $P$. 
    Since pre-normality of a language w.r.t. a plant is equivalent to pre-normality of its complement, we conclude that $L_n\setminus \Psi_f^{-2}$ is neither pre-normal w.r.t. $L(G_1)$ and $P$. By Proposition \ref{Pro:prodef}, $G_1$ is not $2$-prognosable.     
    \hfill$\blacksquare$
\end{example}

\subsection{Characterizations of diagnosability}

When a plant is not $k$-prognosable, we progressively reduce the value of $k$ to check whether prognosability holds for a smaller $k$. If the plant still fails to be prognosable even for $k=0$, the objective shifts to testing fault diagnosability after fault occurrences.
To this end, we focus on the characterization of diagnosability.
Let $L_{\Psi_f}^{\geq N}=\{st\in L_f \mid s\in \Psi(\Sigma_f),\, |P(t)|\geq N\}$ be the set of strings consisting of at least $N\in \mathbb{N}$ observation strings after a fault.

\begin{proposition}\label{Prop: diaPsi_f}
A live and convergent DFA $G$ is diagnosable w.r.t. $\Sigma_f$ and $P$ if and only if there exists $N\in \N$ such that  $P^{-1}P(L_{\Psi_f}^{\geq N})\cap L(G) \subseteq L_f$.
\end{proposition}

\begin{proof}
This proof is straightforward from Definition \ref{Def: dia} and the notion of $L_{\Psi_f}^{\geq N}$.
\end{proof}

By the definition of pre-normality and Proposition \ref{Prop: diaPsi_f}, since $L_{\Psi_f}^{\geq N}\subseteq L_f$, pre-normality of $L_{\Psi_f}^{\geq N}$ is a sufficient condition for diagnosability.
We show that diagnosability is equivalent to the pre-normality of another sublanguage of the faulty language. Let us denote the set of strings from the faulty language consisting of at least $N$ observable events by $L_f^{\geq N}=\{s\in L_f\mid |P(s)|\geq N\}$.  

\begin{proposition}
\label{Prop: ldia}
Assume that $P(\Psi(\Sigma_f))$ is finite. A live and convergent DFA $G$ is diagnosable w.r.t. $\Sigma_f$ and $P$ if and only if there exists $N\in \N$ such that $P^{-1}P(L_{f}^{\geq N})\cap L(G) \subseteq L_{f}^{\geq N}$.
\end{proposition}

\begin{proof}
Let $N'\in \N$ be the largest number of observations before the fault occurs for the first time in $G$, which is finite since $P(\Psi(\Sigma_f))$ is finite. There are $N_1,N_2\in\N$ such that $P^{-1}P(L_{f}^{\geq N'+N_1})\cap L(G) \subseteq L_f$, where $L_{f}^{\geq N'+N_1}=L_{\Psi_f}^{\geq N_2}$. Equivalently, $P^{-1}P(L_{\Psi_f}^{\geq N_2})\cap L(G) \subseteq L_f$, which is by Proposition \ref{Prop: diaPsi_f} equivalent to diagnosability. 
Furthermore, we have $L_f^{\geq N}= L_f \parallel \Sigma_o^{\geq N}$, where $\Sigma_o^{\geq N}=\{t\in \Sigma_o^*\mid |P(t)|\geq N\}$. Since the strings with the same observation have the same number of observable events, it is equivalent to requiring $P^{-1}P(L_{f}^{\geq N})\cap L(G) \subseteq L_{f}^{\geq N}$. 
\end{proof}

By the definition of pre-normality and Proposition \ref{Prop: ldia}, assuming $P(\Psi(\Sigma_f))$ is finite, diagnosability is equivalent to pre-normality of the language $L_f^{\geq N}$  w.r.t. $L$ for some $N\in \N$. 
Even if the assumption does not hold, i.e., $P(\Psi(\Sigma_f))$ is not finite, by Proposition \ref{Prop: diaPsi_f}, a sufficient condition for diagnosability verification exists, i.e., a DFA $G$ is diagnosable w.r.t. $\Sigma_f$ and $P$ if there exists $N\in \N$ such that $P^{-1}P(L_{f}^{\geq N})\cap L(G) \subseteq L_{f}^{\geq N}$. In this way, we can enforce diagnosability by enforcing pre-normality, which allows us to compute diagnosable and pre-normal sublanguages.
  
Since pre-normality and normality coincide for prefix-closed languages and pre-normality is symmetric w.r.t. complement, if $P(\Psi(\Sigma_f))$ is finite, a live and convergent plant $G$ is diagnosable w.r.t. $P$ and $\Sigma_f$ if and only if there is $N\in \N$ such that the language $L_n\cup L_{f}^{< N}$ is normal, i.e., $P^{-1}P(L_n\cup L_{f}^{< N})\cap L \subseteq L_n\cup L_{f}^{< N}$, where $L_f^{< N}=\{s\in \overline{L_f}\mid |P(s)|< N\}$. 

The authors in \cite{Genc09pro} claim that prognosability implies diagnosability. In the following, thanks to Propositions \ref{Pro:prodef} and \ref{Prop: ldia}, we show that prognosability equals diagnosability if there exists certain $N\in\N$ such that $P^{-1}P(L_{f}^{\geq N})\cap L(G) \subseteq L_{f}^{\geq N}$.

\begin{proposition}\label{prop:pro=dia}
    \rm
Let $G$ be a live and convergent DFA and $N_s$ be the smallest number of observations in $\Psi(\Sigma_f)$. Then, prognosability is equivalent to diagnosability if $P^{-1}P(L_{f}^{\geq N_s+1})\cap L(G) \subseteq L_{f}^{\geq N_s+1}$. 
\end{proposition}
\begin{proof}
It is shown that prognosability implies diagnosability in \cite{Genc09pro}. We only need to prove that diagnosability implies prognosability if the condition holds.
As is known, the prognosability is equivalent to $0$-prognosability by Corollary \ref{Corollary:equiv}.
According to Proposition \ref{Pro:prodef}, we have $P^{-1}P(\overline{L_n\setminus \Psi_f^{-0}})\cap L \subseteq L_n\setminus \Psi_f^{-0}$. 
Now, we show that $P^{-1}P(\Psi_f^{-0})\cap L \subseteq \Psi_f^{-0}$ holds if $G$ is prognosable. 
By contrapositive, there exists a string $s\in P^{-1}P(\Psi_f^{-0})$ but $s\notin \Psi_f^{-0}$. By the definition of $\Psi_f^{-0}$, $\Psi(\Sigma_f)\subseteq \Psi_f^{-0}$. It holds that $s\notin \Psi(\Sigma_f)$, i.e., $s\in L_n$. Due to $s\notin \Psi_f^{-0}$, for all $i\in \mathbb{N}$, there exists $t_i\in L/s$ with $|t_i|\geq i$ such that $st_i\in L_n$. By Definition~\ref{Def: pro}, $G$ is not prognosable, which leads to a contradiction.

Due to $\Psi_f^{-0}\nsubseteq L_n$, $(L_n\setminus \Psi_f^{-0}) \cup \Psi_f^{-0}=L_n\cup \Psi_f^{-0}$ holds.
Since the union of two pre-normal languages is also pre-normal w.r.t. $L$, we have $P^{-1}P(L_n\cup \Psi_f^{-0})\cup L \subseteq L_n\cup \Psi_f^{-0}$. By the definition of the languages $\Psi_f^{-0}$ and $L_{f}^{< N}$, we have $L_n\cup \Psi_f^{-0}=L_n\cup L_{f}^{< N_s+1}$. It holds that $P^{-1}P(L_n\cup L_{f}^{< N_s+1})\cap L \subseteq L_n\cup L_{f}^{< N_s+1}$, which completes the proof.
\end{proof}

Proposition \ref{prop:pro=dia} shows that verifying the prognosability of a DFA $G$ does not require computing language $\Psi_f^{-0}$ and checking whether $L_n \setminus \Psi_f^{-0}$ is prefix-closed and pre-normal, while it suffices to test the pre-normality of $L_f^{\geq N_s+1}$.
Further, based on Proposition~\ref{Prop: ldia}, to verify whether $G$ is diagnosable, it is sufficient to check $P^{-1}P(L_{f}^{\geq N})\cap L(G) \subseteq L_{f}^{\geq N}$ for some $N\in \N$. 

\begin{proposition}
\label{prop:max}
Let $G$ be a DFA recognizing $L$ and $Obs(G)$ be its observer, whose state cardinality is $N_o$.  $G$ is diagnosable w.r.t. projection $P: \Sigma^* \to \Sigma_o^*$ and the set of fault events $\Sigma_f$ if and only if $P^{-1}P(L_{f}^{\geq N_o})\cap L(G) \subseteq L_{f}^{\geq N_o}$.
\end{proposition}
\begin{proof}
It is sufficient to show that there exists $ N\in \N$ such that  $P^{-1}P(L_{f}^{\geq N})\cap L(G) \subseteq L_{f}^{\geq N}$ if and only if $P^{-1}P(L_{f}^{\geq N_o})\cap L(G) \subseteq L_{f}^{\geq N_o}$ by Proposition  \ref{Prop: ldia}.
We first show the following monotonicity property, i.e., 
 if $P^{-1}P(L_{f}^{\geq N})\cap L \subseteq L_{f}^{\geq N}$, then $P^{-1}P(L_{f}^{\geq N+1})\cap L \subseteq L_{f}^{\geq N+1}$.
 Since $\Sigma_o^{\geq N+1}=\Sigma_o^{\geq N} \cap \Sigma_o^{\geq N+1}$, we have
\begin{align*}
&  P^{-1}P(L_f^{\geq N+1}) \cap L=P^{-1}P(L_f\cap P^{-1}(\Sigma_o^{\geq N+1}))\cap L\\
&=  P^{-1}P(L_f\cap P^{-1}(\Sigma_o^{\geq N} \cap \Sigma_o^{\geq N+1})) \cap L\\
&\subseteq   P^{-1}P(L_f\cap P^{-1}(\Sigma_o^{\geq N}))\cap L \cap
P^{-1}P(P^{-1}(\Sigma_o^{\geq N+1}))\\
&\subseteq L_f\cap P^{-1}(\Sigma_o^{\geq N})\cap
P^{-1}(\Sigma_o^{\geq N+1})=L_f^{\geq N+1}.
\end{align*}

This means that the pre-normality of $L_f^{\geq N}$ is stronger than that of $L_f^{\geq N+1}$ for all $N\geq 0$. Intuitively, since we deal with finite automata and natural projections as observations (with finite state observers), it should not be surprising that we cannot weaken the pre-normality of these languages indefinitely in this manner, but it will be useless to consider $N$ from some value on.
Now, the size $N_o$ of the observer of $G$ is used to show that $N_o$ is the right value, meaning that it is useless to consider pre-normality of $L_f^{\geq N}$ for $N> N_o$. 

Consider languages $L_n\cup L_f^{< N}$ for different values of $N$.
According to the definition of pre-normality,  $L_f^{\geq N}$ is pre-normal w.r.t. $L$ and $P$ if and only if there do not exist two strings $w_1,w_2\in L(G)$ such that $P(w_1)=P(w_2)$, $w_1\in L_n\cup L_f^{< N}$ and $w_2\in L_f^{\geq N}$. Otherwise, for all $w_2\in L_f^{\geq N}$ and for all $w_1 \in L(G)$ with $P(w_1)=P(w_2)$, we have $w_1\in L_f^{\geq N}$, which is equivalent to diagnosability after $N$ observations.
Notice that according to Proposition~\ref{Prop: ldia}, $L_f^{\geq N}$ is pre-normal w.r.t. $L$ iff $G$ is diagnosable in $N$ observable steps. However, the number of observations needed to diagnose a fault in a diagnosable DFA is upper bounded by $N_o$, i.e., the size of the observer. This ends the proof.
\end{proof}

Propositions \ref{Prop: ldia} and \ref{prop:max} imply that, to verify whether a DFA $G$ is diagnosable w.r.t. a fault $\Sigma_f$, it suffices to check whether the language $L_f^{\geq N_o}$ is pre-normal w.r.t. $L(G)$, where $N_o$ is the number of observer states.
To test the pre-normality of $L_f^{\geq N_o}$, we need to mark the states reached by firing $s'\in L_n\cup L_f^{< N_o}$ and unmark the states reached by firing $s\in L_f^{\geq N_o}$. However, there may exist a state that is both marked and non-marked. To address this issue, we introduce a verifier-based approach to check the pre-normality of $L_f^{\geq N}$.

Building upon the verifier in \cite{masopust2019critical}, we introduce a slight modification, i.e., the transition labels are changed from single events to pairs of seemingly identical events. Additionally, to reduce the computational burden, certain symmetric sequences are avoided. 
In fact, the verifier defined in the following can be regarded as a sub-automaton of the verifier in \cite{masopust2019critical}.

\begin{definition}[Verifier]\label{Def:verifier}
    \rm
    Given a DFA $G=(Q,\Sigma, \delta, q_0)$, its verifier, denoted by $G|||G$, is a DFA $G|||G$ $=(V,\Sigma_{v},\delta_{v}, V_0)$, where $V\subseteq Q\times Q$ is the set of states,  $\Sigma_{v}\subseteq(\Sigma\cup\{\varepsilon\}\setminus\Sigma_f) \times (\Sigma\cup\{\varepsilon\})$ is the event set, $V_0=(q_0,q_0)$ is the initial verifier state, and $\delta_{v}\subseteq V\times \Sigma_v\times V$ is the transition function such that for all $q,q'\in Q$ $\delta_{v}((q,q'),(\alpha,\alpha'))=(\delta(q,\alpha),\delta(q',\alpha'))$ if $\alpha=\alpha'\in \Sigma_o$; $\delta_{v}((q,q'),$ $(\alpha,\varepsilon))=(\delta(q,\alpha), q')$ if $\alpha\in \Sigma_{uo}\setminus \Sigma_f$; $\delta_{v}((q,q'),$ $(\varepsilon,\alpha'))=(q, \delta(q',\alpha'))$ if $\alpha'\in \Sigma_{uo}$.
    \hfill$\Diamond$
\end{definition}

Let $Q_m\subseteq Q$ denote the set of marked states of $G$.
Given a state  $(q,q')=\delta_v((q_0,q_0), (s_1,s_2)) \in V$, where $s_1,s_2\in L$, let $q'\in Q_m$ if $s_2\in L_f^{\geq N}$. Since $s_1\in L_n$, state $q$ is always outside $Q_m$. A state $(q,q')\in V$ is said to be uncertain if $q'\in Q_m$. The set of all uncertain states is defined as $V_c^N=\{(q,q')\in V \mid \exists s\in L_n, \exists s'\in L_f^{\geq N}: \delta_v(V_0, (s, s'))=(q,q')\}$.

\begin{lemma}\label{lemma:verifier}
Let $G$ be a DFA recognizing $L$, and $G|||G$ be its verifier. Language $L_n\cup L_f^{< N}$ is not pre-normal w.r.t. $L$ and natural projection $P$ if and only if there exists an uncertain state $(q,q')\in V_c^N$.
\end{lemma}
\begin{proof}
(If) Let $(q,q')\in V_c^N$ be an uncertain state in $G|||G$. There exist two sequences $s_1,s_2\in L$ with $P(s_1)=P(s_2)$ such that $s_1\in L_n$ and $s_2\in L_f^{\geq N}$. Due to the definition of pre-normality, $L_n\cup L_f^{< N}$ is not pre-normal w.r.t. $L$ and $P$.

(Only if)  
By the definition of pre-normality, if $L_n\cup L_f^{< N}$ is not pre-normal w.r.t. $L$ and $P$, there exist two sequences $w_1,w_2\in L(G)$ with $P(w_1)=P(w_2)$, $w_1\in L_n\cup L_f^{< N}$ and $w_2\in L_f^{\geq N}$. Since $P(w_1)=P(w_2)$, i.e., $w_1$ and $w_2$ have the same number of observations, we necessarily have $w_1\in L_n$. Due to Definition \ref{Def:verifier}, there is a verifier state $(q,q')\in V_c^N$, where $q=\delta(q_0,w_1)$ and $q'=\delta(q_0,w_2)$. We conclude that $(q,q')\in V_c^N$ is an uncertain state in $G_n|||G_n^N$.
\end{proof}

\begin{example} \label{Exa:diaN_o}
    Given a DFA $G_2=(Q_2, \Sigma_2, \delta_2, q_{0,2})$ depicted in Fig. \ref{G_2}(a), the event set is $\Sigma_2=\{a,b,f_2,\lambda\}$ with fault event set $\Sigma_{2,f}=\{f_2\}$, where $a$ and $b$ are observable, and $\lambda$ is non-faulty unobservable. Then $L_n=(\lambda ab^n)^*\cup \{a\}$ and $L_f=L\setminus L_n$. Since there exist a fault-ending string $af_2\in \Psi(\Sigma_f)$ and a non-faulty string $\lambda a\in L(G_2)$ with $|P(af_2)|=|P(\lambda a)|=0$ such that for all $m\in \N$, there is a non-faulty extension $b^m\in \lambda a \backslash L(G_2)$, $G_2$ is not $0$-prognosable w.r.t. $\Sigma_{2,f}$ and $P$ by Lemma \ref{Lem:non-k-pro}. Due to Corollary \ref{Corollary:equiv}, $G_2$ is not prognosable w.r.t. $\Sigma_{2,f}$ and $P$. We then test whether it is diagnosable.
    Since for all $s\in L_f$ with $|P(s)|\geq 3$, there does not exist a string $s'\in L_n$ such that $P(s)=P(s')$, $G_2$ is diagnosable w.r.t. $\Sigma_{2,f}$ and $P$ by Definition~\ref{Def: dia}. 
    The observer and verifier of $G_2$, i.e., $Obs(G_2)$ and $G_2|||G_2$ are shown in Figs. \ref{G_2}(b) and \ref{G_2}(c), respectively. 
      
We test $N=2$ first. By the definition of $V_c^N$, the set of uncertain verifier states is $V_c^{4}=\{(4,0),(4,3)\}$, which are portrayed in red in Fig. \ref{G_2}(c). According to Lemma \ref{lemma:verifier},  $L_n\cup L_f^{< 2}$ is not pre-normal w.r.t. $L(G_2)$ and $P$.
    Then we consider $N=N_o=4$.
    The set of uncertain verifier states is $V_c^{4}=\emptyset$. Due to Lemma \ref{lemma:verifier}, $L_f^{\geq 4}$ is pre-normal w.r.t. $G_2$. 
    By Propositions~\ref{Prop: ldia} and \ref{prop:max}, $G_2$ is diagnosable w.r.t. $P$ and $\Sigma_{2,f}$.  
    \hfill$\blacksquare$
\end{example}

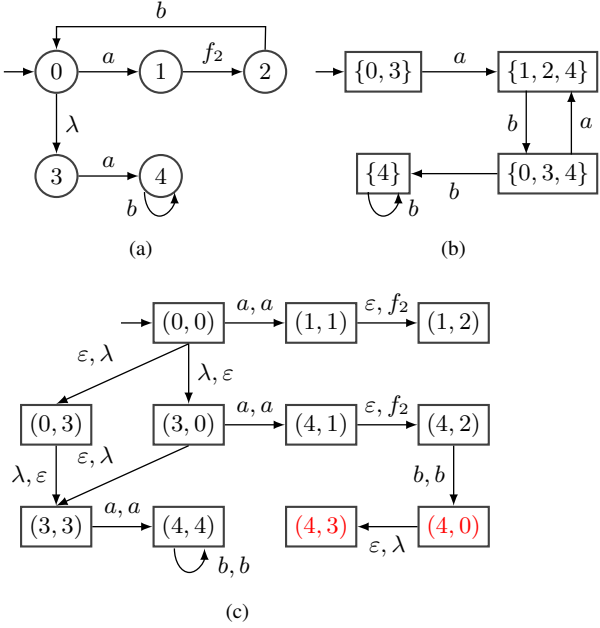
\begin{figure} [htbp]
\subfigure[]
	{
	\centering
	\begin{tikzpicture}{every node}=[font= \fontsize{9pt}{10pt},scale=1]
	\node (c0) [circle,thick,draw=black!75] {$0$};
	\node (c1) [circle,thick,draw=black!75,right=0.8cm of c0] {$1$};
	\node (c2) [circle,thick,draw=black!75,right=0.8cm of c1] {$2$};
	\node (c3) [circle,thick,draw=black!75,below=0.8cm of c0] {$3$};
	\node (c4) [circle,thick,draw=black!75,right=0.8cm of c3] {$4$};

        \draw[-latex,line width=0.5pt] (-0.7,0) -- (c0.west);
	\draw[-latex,line width=0.5pt] (c0.east) --node [auto]{$a$} (c1.west);
	\draw[-latex,line width=0.5pt] (c1.east) --node [auto]{$f_2$} (c2.west);
	\draw[-latex,line width=0.5pt] (c0.south) --node [auto]{$\lambda$} (c3.north);
	\draw[-latex,line width=0.5pt] (c3.east) --node [auto]{$a$} (c4.west);
    \draw[-latex,line width=0.5pt] (c2.north) -- (2.78,0.6)--node [auto,swap]{$b$}(0,0.6)-- (c0.north);
	\draw[-latex,line width=0.5pt] (1.17,-1.6)arc[start angle=-180, end angle=0, x radius=.2, y radius=.32];
	\node[black,below] at (1,-1.58) {{$b$}};
	\end{tikzpicture}	}
\subfigure[]
{
	\centering
      \begin{tikzpicture}{every node}=[font= \fontsize{9pt}{10pt},scale=1]
		\node (r0) [draw,thick,draw=black!75] {$\{0,3\}$};
		\node (r1) [draw,thick,draw=black!75,right=1cm of r0] {$\{1,2,4\}$};
        \node (r2) [draw,thick,draw=black!75,below=0.8cm of r1] {$\{0,3,4\}$};
        \node (r3) [draw,thick,draw=black!75,below=0.8cm of r0] {$\{4\}$};

    \draw[-latex,line width=0.5pt] (-0.9,0) -- (r0.west);
    \draw[-latex,line width=0.5pt] (r0.east) --node [auto]{$a$} (r1.west);
    \draw[-latex,line width=0.5pt] (r2.west) --node [auto]{$b$} (r3.east);
    \draw[-latex,line width=0.5pt] (1.9,-0.25) --node [auto,swap]{$b$} (1.9,-1.1);
    \draw[-latex,line width=0.5pt] (2.5,-1.1) --node [auto,swap]{$a$} (2.5,-0.25);
	\draw[-latex,line width=0.5pt] (-0.2,-1.6)arc[start angle=-180, end angle=0, x radius=.2, y radius=.32];
	\node[black,below] at (.4,-1.58) {{$b$}};
\end{tikzpicture}	}

\subfigure[]
{
    \centering
    \begin{tikzpicture}{every node}=[font= \fontsize{9pt}{10pt},scale=1]
    \node (r0) [draw,thick,draw=black!75] {$(0,0)$};
    \node (r1) [draw,thick,draw=black!75,right=0.8cm of r0] {$(1,1)$};
    \node (r2) [draw,thick,draw=black!75,right=0.8cm of r1] {$(1,2)$};
    \node (r3) [draw,thick,draw=black!75,below=0.8cm of r0] {$(3,0)$};
    \node (r4) [draw,thick,draw=black!75,left=0.8cm of r3] {$(0,3)$};
    \node (r5) [draw,thick,draw=black!75,below=0.8cm of r1] {$(4,1)$};
    \node (r6) [draw,thick,draw=black!75,below=0.8cm of r2] {$(4,2)$};
    \node (r7) [draw,thick,draw=black!75,below=0.8cm of r3] {$(4,4)$};
    \node (r8) [draw,thick,draw=black!75,below=0.8cm of r4] {$(3,3)$};
    \node (r9) [draw,thick,draw=black!75,below=0.8cm of r5] {\textcolor{red}{$(4,3)$}};
    \node (r10) [draw,thick,draw=black!75,below=0.8cm of r6] {\textcolor{red}{$(4,0)$}};

    \draw[-latex,line width=0.5pt] (-0.9,0) -- (r0.west);
    \draw[-latex,line width=0.5pt] (r0.east) --node [auto]{$a,a$} (r1.west);
    \draw[-latex,line width=0.5pt] (r1.east) --node [auto]{$\varepsilon,f_2$} (r2.west);
    \draw[-latex,line width=0.5pt] (r3.east) --node [auto]{$a,a$} (r5.west);
    \draw[-latex,line width=0.5pt] (r5.east) --node [auto]{$\varepsilon,f_2$} (r6.west);
    \draw[-latex,line width=0.5pt] (r8.east) --node [auto]{$a,a$} (r7.west);
    \draw[-latex,line width=0.5pt] (r10.west) --node [auto]{$\varepsilon,\lambda$} (r9.east);
    \draw[-latex,line width=0.5pt] (r0.south) --node [auto]{$\lambda,\varepsilon$} (r3.north);
    \draw[-latex,line width=0.5pt] (r0.south) --node [auto,swap]{$\varepsilon,\lambda$} (r4.north);
    \draw[-latex,line width=0.5pt] (r4.south) --node [auto,swap]{$\lambda,\varepsilon$} (r8.north);
    \draw[-latex,line width=0.5pt] (r3.south) --node [auto,swap]{$\varepsilon,\lambda$} (r8.north);
    \draw[-latex,line width=0.5pt] (r6.south) --node [auto,swap]{$b,b$} (r10.north);

	\draw[-latex,line width=0.5pt] (-0.19,-3)arc[start angle=-180, end angle=0, x radius=.2, y radius=.32];
	\node[black,below] at (0.6,-3) {{$b,b$}};
\end{tikzpicture}}
    \caption{(a) A DFA $G_2$, (b) its observer $Obs(G_2)$, and (c) its verifier $G_2|||G_2$.}
    \label{G_2}
\end{figure}

\section{Existence of supremal prognosable/diagnosable and normal languages}\label{Sec: Existence}

In this section, we consider the case where $G$ fails to be $k$-prognosable (resp. diagnosable), and we are looking for the largest possible sublanguages of the plant that satisfy these properties. 
We will show that supremal $k$-prognosable (resp. diagnosable) and normal sublanguages always exist. 
We need the following result stating the transitivity of pre-normality.

\begin{lemma}\label{trans_LCO}
    Let $L,L',$ and $M$ be languages such that  $L'\subseteq L\subseteq M\subseteq \Sigma^*$, $L'$ is pre-normal w.r.t. $L$ and $P$, and $L$ is pre-normal w.r.t. $M$ and $P$. Then, $L'$ is pre-normal w.r.t. $M$ and $P$.
\end{lemma}
\begin{proof}
    We know that $P^{-1}P(L')\cap L = L'$ and $P^{-1}P(L)\cap M = L$. Then, $P^{-1}P(L')\cap M \subseteq P^{-1}P(L)\cap M = L$. This implies that $P^{-1}P(L')\cap M = (P^{-1}P(L')\cap M) \cap L = (P^{-1}P(L')\cap L) \cap M = L' \cap M = L'$, i.e. $L'$ is pre-normal w.r.t. $M$ and $P$, which completes the proof.
 \end{proof}

 \begin{proposition}\label{lemma:punion}
Let $L(G)=L_f \cup L_n$ be the language recognized by DFA  $G$ with the corresponding faulty and non-faulty languages.
Let $L_i\subseteq L(G)$, where $i\in I$, be a family of sublanguages that are $k$-prognosable and normal w.r.t. $L$ and $P$. Then their union $\cup_{i\in I} L_i$ is also  $k$-prognosable and normal w.r.t. $L$ and $P$.  
\end{proposition}

\begin{proof}
Since it is well known that $\cup_{i\in I} L_i$ is   normal w.r.t. $L$ and $P$, let us show that $\cup_{i\in I} L_i$ is $k$-prognosable. 
Denote by $L_i=L_{i,f}\cup L_{i,n}$ the decomposition of $L_i$'s into their non-faulty and faulty parts.
Let us denote $\Psi_{i,f}^{-k}$ the set of strings that may be extended to reach a fault in at most $k\in \N$ observations and contain all their faulty extensions.  

By Proposition \ref{Pro:prodef}, it suffices to show that $\cup_{i\in I} \overline{L_{i,n}\setminus \Psi_{i,f}^{-k}}$ is pre-normal w.r.t. $\cup_{i\in I} L_i$, i.e.
$P^{-1}P(\overline{\cup_{i\in I} L_{i,n}\setminus \Psi_{i,f}^{-k}})\cap L \subseteq \cup_{i\in I} L_{i,n}\setminus \Psi_{i,f}^{-k}$.
For simplicity we write in this proof  $L_{i,n\setminus f}^{-k}=L_{i,n}\setminus \Psi_{i,f}^{-k}$ and  $L_{n\setminus f}^{-k}=L_{n}\setminus \Psi_{f}^{-k}$.

It amounts to showing that
\begin{equation} \label{eq:punion}
P^{-1}P (\cup_{i\in I} L_{i,n\setminus f}^{-k} )\cap [\cup_{i\in I} L_i]  \subseteq  \cup_{j\in I} L_{j,n\setminus f}^{-k} .
\end{equation}
After distributing the first union with $P^{-1}P$ and distributing both unions with intersection, we obtain by distinguishing terms with $i=j$:
\begin{multline*}
 P^{-1}P(\cup_{i\in I} L_{i,n\setminus f}^{-k}) \cap [\cup_{i\in I} L_i] \\
 = \cup_{i\in I} P^{-1}P(L_{i,n\setminus f}^{-k}) \cap L_i  \cup  \bigcup_{j\not=i}
L_{i,n\setminus f}^{-k} \cap L_j
\end{multline*}

Note that $(\cup_{i\in I} L_{i,n\setminus f})^{-k}=\cup_{i\in I} L_{i,n\setminus f}^{-k}$, where
$L_{i,n\setminus f}^{-k}= L_i \cap L_{n\setminus f}^{-k} $.
Thus,  to prove inequality (\ref{eq:punion}), it suffices to show that
 the mixed terms do not increase the language on the left-hand side. Namely, for every $i,j\in I, j\not=i$, let us consider the languages
$ P^{-1}P \bigl[L_{i,n\setminus f}^{-k})  \bigr]  \cap  L_j$. 
We now use the assumption that $L_i$ is normal w.r.t. $L$, which for prefix-closed languages means
$ P^{-1}P(L_i) \cap L \subseteq L_i$, the distributivity of projections and inverse projections w.r.t. unions,
and the fact that $L_j \subseteq L$ to get:
\begin{align*} 
    & P^{-1}P ( L_i\cap L_{n\setminus f}^{-k} )  \cap  L_j 
    \subseteq  P^{-1}P (L_i)  \cap  P^{-1}P ( L_{n\setminus f}^{-k}) \cap L \\
    & = L_i \cap   P^{-1}P ( L_{n\setminus f}^{-k}). 
\end{align*}
Note that 
by transitivity of pre-normality, cf. Lemma \ref{trans_LCO}, we obtain from 
$L_i \cap   P^{-1}P ( L_{n\setminus f}^{-k})$ is pre-normal w.r.t. $L_i$  and  $L_i$ is pre-normal w.r.t. $L$ (normality of prefix closed $L_i$ w.r.t. $L$) that $L_i \cap   P^{-1}P ( L_{n\setminus f}^{-k})$ is pre-normal w.r.t. $L$, that is,
$$P^{-1}P (L_i \cap   P^{-1}P ( L_{n\setminus f}^{-k})) \cap L
\subseteq L_i \cap   P^{-1}P ( L_{n\setminus f}^{-k}).$$
Altogether, 
$$ P^{-1}P (L_i \cap   P^{-1}P ( L_{n\setminus f}^{-k}) ) \cap  L_j \subseteq L_i \cap   P^{-1}P ( L_{n\setminus f}^{-k}).$$
Therefore, inequality \eqref{eq:punion} holds and
 $\cup_{i\in I} L_i \cap   P^{-1}P ( L_{n\setminus f}^{-k})$ is pre-normal w.r.t. $\cup_{i\in I} L_i$. By Proposition \ref{Pro:prodef}, we conclude that $\cup_{i\in I} L_i$ is $k$-prognosable and normal w.r.t. $L$ and $P$.
\end{proof}

It follows from Proposition \ref{lemma:punion} that the supremal $k$-prognosable sublanguage of $L$ that is normal w.r.t. $L$ always exists and equals the union of all sublanguages of  $L$ that are $k$-prognosable and  normal w.r.t. $L$.
We denote it by ${\rm supNP}^k(L,\Psi_f^{-k}, P)= \{ \cup_{i\in I} L_i \mid$ $L_i \text{ is  $k$-prognosable and normal w.r.t. $L$ and $P$} \}$. 
\begin{remark}
We emphasize that, unlike standard notation for supremal languages in supervisory control, where the specification comes first and the plant in the second place, here the plant comes first because we are looking for the largest sublanguage of the plant that is $k$-prognosable and normal. The language $\Psi_f^{-k}$ then plays the role of the specification, because $k$-prognosability is equivalent to pre-normality of it w.r.t. a plant (Proposition \ref{Pro:prodef}). However, in active prognosis, we do not compute the (supremal) sublanguage of the specification $L_n\setminus \Psi_f^{-k}$ like in classical supervisory control theory, but naturally rather the sublanguage of the plant. Note also that this new $k$-prognosable (sub)-plant needs to be normal w.r.t. the original plant $L$ anyway (to be achievable by a supervisor), hence this normality is not an additional restriction. \hfill$\blacksquare$
\end{remark}

A similar result holds for diagnosability, namely that the supremal normal and diagnosable sublanguage exists, as also shown based on a game-theoretic approach by Yin and Lafortune~\cite{yin2015uniform}. 

\begin{proposition}\label{lemma:union}
Let $L(G)=L_f \cup L_n$ be the language recognized by DFA  $G$ with the corresponding faulty and non-faulty languages.
Let $L_i\subseteq L(G)$, where $i\in I$, be a family of sublanguages that are diagnosable w.r.t. $\Sigma_f$ and normal w.r.t. $L$ and $P$. Then their union $\cup_{i\in I} L_i$ is also  diagnosable and normal w.r.t. $L$ and $P$.  
\end{proposition}
\begin{proof}
    Straightforward from Propositions \ref{Prop: ldia} and \ref{lemma:punion}.
\end{proof}

If a fault cannot be diagnosed after $N_o$ observations, then we compute the supremal diagnosable (w.r.t. $N_o$ observations) sublanguage that is normal w.r.t. $L$. We denote the supremal  diagnosable sublanguage of $L$ that is normal w.r.t. $L$ by ${\rm supND}(L,L_{f}^{\geq N_o}, P)$.

\section{Active prognosis and diagnosis}\label{Sec: Active diagnosis}

In this section, we will present an approach to enforce prognosability (resp. diagnosability). There are several approaches in the literature for enforcing prognosability (resp. diagnosability). Notably, one approach is based on a diagnoser or a verifier, which involves removing indeterminate cycles that violate prognosability (resp. diagnosability) through supervisory control \cite{haar2020active,Paoli2005supervisor, yin2015uniform,hu2025iot,Hu2020supervisor, Hu2023dt}. Other approaches focus on sensor selection \cite{Ranproenforcement2022,Ran2019sensor,Hu2024sensor}.
As far as we know, there are few works that touch upon computing the supremal prognosable (resp. diagnosable) and normal sublanguage. 

In Section \ref{Sec: Diagnosability}, we characterized (i) prognosability in terms of the pre-normality of a sublanguage of the non-faulty language  
(in Proposition \ref{Pro:prodef}) and (ii) diagnosability in terms of pre-normality of 
an extension of the non-faulty language by a prefix of the faulty language, determined by a bounded number of observable events (in Proposition \ref{prop:max}). In Section~\ref{Sec: Existence}, it is shown that the supremal $k$-prognosable and normal sublanguage as well as the supremal diagnosable and normal sublanguage exist by Propositions \ref{lemma:punion} and \ref{lemma:union}, respectively.

\subsection{Active prognosis}

Let us recall that the supremal normal and $k$-prognosable sublanguage ${\rm supNP}^k(L,L_n \setminus \Psi_f^{-k},P)$ always exists.
Note that  by Proposition \ref{Pro:prodef} we need not only that $L_{n}\setminus \Psi_{f}^{-k}$ is pre-normal w.r.t. a subplant $L'\subseteq L$, but also that  $L'_{n}\setminus \Psi_{f}^{-k}=(L_{n}\cap L')\setminus \Psi_{f}^{-k}$ is prefixed-closed to achieve $k$-prognosability of $L'$. 
Since $L'_{n}\setminus \Psi_{f}^{-k}$ is not always prefix-closed, we need to compute $L''\subseteq L'\subseteq L$ such that   $(L'_{n}\cap L'')\setminus \Psi_{f}^{-k}$ is  prefix-closed. First, we show that prefix-closedness is preserved by enforcing normality.

\begin{lemma} \label{Lemma:profix}
   Let  $L(G)$ be the language recognized by a  non $k$-prognosable DFA $G$ that satisfies Assumption A1 with $L(G)=L_f \cup L_n$. Let $L'\subseteq L$ with $L'=\overline{L'}$ be the sublanguage of $L$ such that $L'_{n}\setminus \Psi_{f}^{-k}=(L_n\cap L')  \setminus \Psi_{f}^{-k}$ is now prefix-closed w.r.t. $L'$ and $L'_f=L_f\cap L'$. If there exists a sublanguage $L''\subseteq L'$ with $L''=\overline{L''}$ such that $L''_{n}\setminus \Psi_{f}^{-k}=(L_n\cap L'')  \setminus \Psi_{f}^{-k} $ is normal w.r.t. $L'$, then $L''_{n}\setminus \Psi_{f}^{-k}$ is also prefix-closed w.r.t. $L''$.
\end{lemma}
\begin{proof}
The condition holds if $L''=L'$. 
Now we consider $L''\subsetneq L'$. According to Proposition \ref{Pro:prodef}, if $L'_{n}\setminus \Psi_{f}^{-k}$ is prefix-closed but not normal, for every string $s\in L'_{n}\setminus \Psi_{f}^{-k}$ and $t\in L'$ with $P(s)=P(t)$ but $t\notin L'_{n}\setminus \Psi_{f}^{-k}$, there does not exist a string $t'\in L'/t$ such that $tt'\in L_f'$. In other words, for all $t'\in L'/t$, $tt'\in L_n'$. If there exists a sublanguage $L''\subseteq L'$ with $L''=\overline{L''}$ such that $L''_{n}\setminus \Psi_{f}^{-k} $ is normal, then either $s\notin L''_{n}\setminus \Psi_{f}^{-k}$ or $t\notin L''$ holds due to $P(s)=P(t)$. Further, by $L''=\overline{L''}$, for all $s'\in L'/s$ and for all $t'\in L'/t$, we have either $s'\notin L''$ or $t'\notin L''$. In summary, we conclude that $L''_{n}\setminus \Psi_{f}^{-k}$ is prefix-closed w.r.t. $L''$.
\end{proof}

Lemma \ref{Lemma:profix} implies that normality does not compromise prefix-closedness for $L_{n}\setminus \Psi_{f}^{-k}$, i.e., if a sublanguage $L_{n}\setminus \Psi_{f}^{-k}$ is prefix-closed, then enforcing $k$-prognosability (normality) preserves this property. 
Before establishing an algorithm to compute the supremal controllable, normal, and $k$-prognosable sublanguage of $G$, some notions and notations are proposed first.

A DFA $H= (X, \Sigma, \Delta, x_0)$ is said to be a \emph{strict sub-automaton} of $G$, denoted by $H \sqsubset G$, if the following conditions hold: 1) $\Delta(x_0, s) = \delta(q_0, s)$ for all $s \in L(H)$, and 2) if $x, x' \in X$ and $\delta(x, s) = x'$ for $s \in \Sigma^*$, then $\Delta(x, s) = x'$.
We abuse the notation and recast the partial transition function $\Delta: X\times \Sigma \to X $ as a subset $\Delta \subseteq  X\times \Sigma \times X$ with an obvious correspondence.
An automaton $G$ is a \emph{state-partition automaton} (SPA) w.r.t. $P$ if any two states of its observer do not have a nontrivial overlap, i.e., either they are identical or their intersection is empty. Without loss of generality, we assume that $H \sqsubset G$, which can be ensured by refining the state space \cite[Section 2.3.3]{cassandras2021introduction}, such that $G$ is an SPA w.r.t. $P$. The SPA property can always be achieved by computing $G \parallel \mathrm{Obs}(G)$ \cite[Section 3.7.5]{cassandras2021introduction} as a new recognizer for $L$.

Given a DFA recognizing $L$, the sets of fault events $\Sigma_f$ and uncontrollable events $\Sigma_{uc}$, and  natural projection $P$, the supremal controllable, normal, and $k$-prognosable sublanguage of $G$ is denoted by $\mathrm{supCNP}^k(L, \Psi_f^{-k}, \Sigma_{uc}, P)$.
The following statement simplifies the computation of $\mathrm{supCNP}^k(L, \Psi_f^{-k}, \Sigma_{uc}, P)$.

\begin{remark}
\label{rem2}
Let us recall at this point the notion of critical observability, restricted to a DFA, which is closely related to $k$-prognosability. A DFA $G = (Q,\Sigma,\delta,q_0)$ is critically observable w.r.t. the set of critical states $Q_c\subseteq Q$ if for every $s,s' \in L(G)$ with $P(s)=P(s')$, $\delta(q_0,s) \in Q_c$ if and only if $\delta(q_0,s') \in Q_c$ \cite{masopust2019critical}. A sub-plant $L'=\overline{L'}\subseteq L(G)$ is critically observable w.r.t. $G$, $P$, and $Q_c$ if, for every $s,s' \in L'$ with $P(s)=P(s')$, $\delta(q_0,s) \in Q_c$ if and only if $\delta(q_0,s') \in Q_c$ \cite{arxiv_CO}.  It is now easy to see that $L'\subseteq L(G)$ is critically observable w.r.t. $G$, $P$, and $Q_c$ if and only if $P^{-1}P(K_c\cap L') \cap L' \subseteq K_c\cap L'$, where $K_c$ is the language corresponding to $Q_c$, i.e. $K_c=\{s\in L(G) \mid \delta(q_0,s)\in Q_c\}$.
It follows from Proposition \ref{Pro:prodef} applied to $L'\subseteq L$ that $L'\subseteq L(G)$ is $k$-prognosable if and only if $L'$ is critically observable w.r.t. $G$, $P$, and the set of critical states given by the critical language $K_c=(L_n\cap L')\setminus \Psi_f^{-k}$ and $(L_n\cap L')\setminus \Psi_f^{-k}$ is prefix-closed.
Note that the prefix-closedness is achieved by computing the supremal prefix-closed sublanguage, and due to Lemma \ref{Lemma:profix}, computation of $k$-prognosable sublanguages $L'\subseteq L(G)$ does not alter prefix-closedness of $(L_n\cap L')\setminus \Psi_f^{-k}$.
\hfill$\blacksquare$
\end{remark}

	\begin{algorithm}[t]
    \caption{Computation of $\mathrm{supCNP}^k(L,\Psi_f^{-k}, \Sigma_{uc}, P)$}
    \label{Alg: pro}
    \begin{algorithmic}[1]
    \REQUIRE A non-negative integer $k\in \{0,1,2,...\}$ and a DFA $G = (Q, \Sigma, \delta, q_0)$, which is an SPA w.r.t. $P$;\\
    a DFA $H_1= (X_1, \Sigma, \Delta_1, x_0)$ with $H_1 \sqsubset G$ and $L(H_1)=L(G)=L$.
    \ENSURE $\mathrm{supCNP}^k(L,\Psi_f^{-k}, \Sigma_{uc}, P)$.
    \STATE Construct the observer $\mathrm{Obs}(H_1) = (Y, \Sigma_o, \xi, y_{0})$~\cite{cassandras2021introduction}; 
        
    \STATE prefix-closedness 
    \[
M: = \underbrace{\{s\in L_n\setminus \Psi_f^{-k} \mid \forall s'\in \overline{s}: s'\in L_n\setminus \Psi_f^{-k}\}}_{\substack{\text{The supremal prefix-closed subset of $L_n\setminus \Psi_f^{-k}$.}}};
\]
    
    \STATE $X_m: = \cup_{s \in M}\{\Delta_1(x_0,s)\}$;

\STATE  pre-normality part of $k$-prognosability
   \[
\underbrace{
\begin{aligned}
X^{N_p} := \{ x \in X_1 \mid\ \forall y \in Y,\ x \in y \Rightarrow {} & y \subseteq X_m \lor\  \\
& y \subseteq X_1 \setminus X_m \};
\end{aligned}
}_{\substack{\text{The set of all $x$ from $X_1$ such that whenever}\\
\text{$x$ belongs to $y$, then $y \subseteq X_m$ or $X_1 \setminus X_m$.}}}
\]
    
\STATE $X_2:=X^{N_p}$, $\Delta_2:=\Delta_1\cap  X_2 \times \Sigma \times X_2$, compute $H_2 = (X_2, \Sigma, \Delta_2, x_0)$, and $i:=2$;
    
    \STATE compute conditions for:

    \begin{itemize}

        \item controllability 
        \[
\underbrace{
\begin{aligned}
X^{C}_i: = X_i \setminus \{ x \in X_i \mid \exists s \in \Sigma^*_{uc} :~ & \delta(x,s)\in Q \land\  \\
& f_i(x,s) \notin X_i \};
\end{aligned}
}_\text{The states satisfying controllability.}
\]

        \item normality 
        \[
            X^{N}_i: = \underbrace{\{ x \in X_i \mid \forall y\in Y, x \in y \Rightarrow y \subseteq X_i \}}_{\substack{\text{The set of all $x$ from $X_i$ such that whenever}\\\text{$x$ belongs to $y$ then $y$ is a subset of $X_i$.}}};
        \]
    \end{itemize}

    \STATE $X'_i := X^{C}_i \cap X^{N}_i$ and  
        $\Delta'_i := \Delta_i \cap  X'_i \times \Sigma \times X'_i$; 

    \STATE $X_{i+1}:=X_i'\setminus \{x\in X_i'\; | \; \nexists s\in\Sigma^*: \Delta_i'(x_0,s)=x\}$ and\\
    $\Delta_{i+1}:=\Delta_i'\cap X_{i+1}\times \Sigma \times X_{i+1}$;

    \IF{$L(H_{i+1})=\emptyset$,}
    \STATE Output: No solution;
    \ELSE
    \IF{$X_{i+1}=X_i$ and $\Delta_{i+1} = \Delta_{i}$,} 
     \STATE Output: $\mathrm{supCNP}^k(L,\Psi_f^{-k}, \Sigma_{uc}, P)  = L(H_{i+1})$; 
    \ELSE
        \STATE $i: = i+1$ and goto Step 6.
    \ENDIF
    \ENDIF
    \end{algorithmic}  
	\end{algorithm}


In Algorithm \ref{Alg: pro}, given a DFA $G=(Q, \Sigma, \delta, q_0)$ that is an SPA w.r.t. $P$, we will iteratively build subautomata $H_i= (X_i, \Sigma, \Delta_i, x_0)$ ($i\in \{1,2,\ldots\}$) of $G$ starting from $L(H_1)=L(G)$, which corresponds to a fix-point procedure that restricts a plant to a subplant that satisfies $k$-prognosability (cf. Proposition \ref{Pro:prodef}) from which follows that $G$ is $k$-prognosable w.r.t. $\Sigma_f$ and $P$ if and only if $L_n\setminus \Psi_f^{-k}$ is prefix-closed and pre-normal w.r.t. $L(G)$.
Initially, we construct the observer $Obs(H_1)$ (Step 1). 
Since language $L_n \setminus \Psi_f^{-k}$ may not be prefix-closed, we need to compute the supremal prefix-closed subset $M\subseteq L_n \setminus \Psi_f^{-k}$ (Step~2). According to Step 3, the set of states reached by firing strings in $M$ is obtained and viewed as the set of marked states, i.e., $X_m=\{x\in X_1\mid \exists s\in M : \Delta_1(x_0,s)=x\}$. According to $Q_m$, we compute in Step 4 the set of states $X^{N_p}\subseteq X_1$, which correspond to a sublanguage of the plant (say $L'\subseteq L(G)$) such that $M$ is pre-normal w.r.t. $L'$, or, equivalently, $L'$ is $k$-prognosable. Although in general there is no such supremal $k$-prognosable  $L'\subseteq L$ (without requiring normality of $L'\subseteq L$), the language given by $X^{N_p}$ computed in Step 4 is unique, and the subsequent iterative computation of its supremal controllable and normal sublanguage in Steps~5--15 gives finally
$\mathrm{supCNP}^k(L,\Psi_f^{-k}, \Sigma_{uc}, P)$.
Specifically, let sub-automaton $H_2 = (X_2, \Sigma, \Delta_2, x_0)$, where $X_2:=X^{N_p}$ and $\Delta_2:=\Delta_1$ (Steps 4--5). 
Then, we compute the intersection of the sets of controllable states $X_i^C$ and normal states $X_i^N$ in the $i$-th iteration, i.e., $X'_i := X^{C}_i \cap X^{N}_i$ (Steps 6--7). After removing the set of unreached states and the corresponding arcs, due to Step 8, a new sub-automaton $H_{i+1}= (X_{i+1}, \Sigma, \Delta_{i+1}, q_0)$ is obtained. If $L(H_{i+1})$ is empty, Algorithm \ref{Alg: pro} returns no solution. If $L(H_{i+1})$ is non-empty, we further test whether $X_{i+1}=X_i$ and $\Delta_{i+1} = \Delta_{i}$ hold. If it is true, Algorithm \ref{Alg: pro} returns the supremal controllable, normal, and $k$-prognosable sublanguage of $G$, i.e., $\mathrm{supCNP}^k(L,\Psi_f^{-k}, \Sigma_{uc}, P)  = L(H_{i})$; Algorithm \ref{Alg: pro} tests the $(i+1)$-th iteration, otherwise.
Algorithm~\ref{Alg: pro} stops after a finite number of iterations. 

Now, we analyze the computational complexity of Algorithm \ref{Alg: pro} in detail. Due to the subset construction, the observer $Obs(G)$ contains at most $2^{|Q|}$ states, resulting in an exponential complexity, i.e., $O(2^{|Q|})$. In Step 2, since $\Psi_f^{-k}=P^{-1}P[ \Psi(\Sigma_f) / P^{-1}(\Sigma_o^{\leq k})]\Sigma^* \cap \overline{\Psi(\Sigma_f)}$, cf. Eq. (\ref{(1111)}),  each component is regular, all operations are closed under regular languages and can be performed in polynomial time w.r.t. the sizes of the underlying automata. 
Consequently, computing the supremal prefix-closed sublanguage $M$ and its corresponding states reached by firing $s\in M$ also remains polynomial, as they involve standard state pruning techniques on DFAs. At each iteration (Steps 5–15), the algorithm checks the conditions of controllability and normality by performing set-based operations over the observer states, whose size is bounded by $O(2^{|Q|})$. The computational complexity of these steps, including checking set inclusions and removing unreachable states, is $O(2^{|Q|}\cdot |Q|)$. Since at each iteration the state space is strictly reduced unless a fixed point is reached, the number of iterations is at most $|Q|$. Therefore, the overall worst-case complexity of the algorithm is $O(2^{|Q|}\cdot |Q|^2)$.
We now obtain the main result of this paper.

\begin{theorem}
\label{Theorem:sup_prognosis_l}
  Let  $L(G)$ be the language recognized by DFA $G$ satisfying Assumption A1, with the partition of faulty and non-faulty sublanguages, i.e., $L(G)=L_f \cup L_n$.  
  The supremal controllable, normal, and $k$-prognosable sublanguage w.r.t. $L(G)$ and projection $P$ is $\mathrm{supCNP}^k(L, \Psi_f^{-k}, \Sigma_{uc},P)$, computed by Algorithm \ref{Alg: pro}.
\end{theorem}
\begin{proof}
In accordance with Proposition \ref{Pro:prodef} and Corollary~\ref{Corollary: prenormal and closed}, we compute the supremal prefix-closed sublanguage of $L_{n}\setminus \Psi_{f}^{-k}$ in Step~2 of Algorithm \ref{Alg: pro}. 
According to Lemma~\ref{Lemma:profix}, prefix-closedness is preserved under controllability and normality computations in Steps 5--15.
Note that, directly from the definition, any sublanguage of a $k$-prognosable plant is also $k$-prognosable. 
Therefore, it is not necessary to iterate $k$-prognosability.
In this way, Algorithm \ref{Alg: pro} first computes, in Step 4, a special $k$-prognosable sublanguage $(L_n\cap L')  \setminus \Psi_{f}^{-k} $. This is the same as \cite[Lemma 6]{arxiv_CO}, since, by Proposition \ref{Pro:prodef} and Remark~\ref{rem2}, $k$-prognosability of the computed subplant $H_2$ (given by $X^{N_p}$ from Step 4) with $L(H_2)\subseteq L$ is equivalent to the critical observability of $H_2$ w.r.t. $X^{N_p}$ as the set of critical states. 
In this way, we can apply the computation scheme from \cite{arxiv_CO} that first computes a special critically observable sublanguage, here $k$-prognosable sublanguage in Step 4 of Algorithm \ref{Alg: pro}, and then controllability and normality are enforced by iterations in Steps 5--15. The supremality then follows \cite[Lemma 7]{arxiv_CO} and Remark \ref{rem2}, namely that after applying the supremal normal operation on the $k$-prognosable sublanguage given by $X^{N_p}$ in Step 4 will always be larger than or equal to the supremal normal operation applied to any other $k$-prognosable sublanguage of the plant.
We conclude that $\mathrm{supCNP}^k(L, \Psi_f^{-k}, \Sigma_{uc},P)$ of Algorithm \ref{Alg: pro} is the supremal $k$-prognosable language that is controllable and normal w.r.t. $L$ and $P$.
\end{proof}

\begin{example} \label{Exa:compute_suppro}
Consider again the system $G_1$ depicted in Fig.~\ref{Fig:G_1}(a). Let $H_1=G_1$. The observer $Obs(G_1)=Obs(H_1)$ is portrayed in Fig.~\ref{Fig:G_1}(b). For simplicity, assume that $E_{c}=E_o$ and $E_{uc}=E_{uo}$. By Example 1, $G_1$ is not $2$-prognosable. Now, we illustrate the computation of $\mathrm{supCNP}^2(L(G_1), \Psi_f^{-2},$ $\Sigma_{uc},P)$ for $G_1$. Specifically, we have $\overline{L_n\setminus \Psi_f^{-2}}= L_n\setminus \Psi_f^{-2} =\{\varepsilon,\tau, \tau a, \tau aa, \tau aabc^j\}$ for $j\in \mathbb{N}$. 
In this way, the set of marked states is $X_m=\{1,2,3,4,9,10\}$. 
According to the observer $Obs(H_1)$, the observer state that contains both marked and unmarked states is $\{1,6\}$. It holds that $X^{N_p}=X_1\setminus\{1,6\}=\{0,2,3,4,5,7,8,9,10\}$. In the first iteration, we compute the condition $X_2^C=X_2^N=\{0,2,3,4,5,7,$ $8,9,10\}$. After removing the unreached states and the corresponding arcs, we have $X_3= \{0,2,3,4,5,9,10\}$ and $L(H_3)$ as shown in Fig.~\ref{fig:supG1}. In the next iteration,  we derive $X_3^C=X_3^N=\{0,2,3,4,5,9,10\}=X_3$. Since the computed sets remain unchanged, we conclude that $\mathrm{supCNP}^2(L(G_1), \Psi_f^{-2}, \Sigma_{uc},P)=L(H_3)$.
    \hfill$\blacksquare$
\end{example}

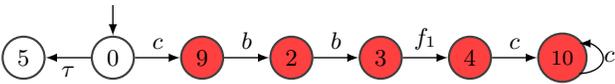
\begin{figure}[htbp]
    \centering
    \begin{tikzpicture}{every node}=[font= \fontsize{9pt}{10pt},scale=1]
	\node (c0) [circle,thick,draw=black!75] {$0$};
	\node (c1) [circle,thick,draw=black!75,fill=red!75,right=0.6cm of c0] {$9$};
	\node (c2) [circle,thick,draw=black!75,fill=red!75,right=0.6cm of c1] {$2$};
    \node (c3) [circle,thick,draw=black!75,fill=red!75,right=0.6cm of c2] {$3$};
	\node (c4) [circle,thick,draw=black!75,fill=red!75,right=0.6cm of c3] {$4$};
	\node (c5) [circle,thick,draw=black!75,left=0.6cm of c0] {$5$};
	\node (c6) [circle,thick,draw=black!75,fill=red!75,font=\footnotesize,right=0.6cm of c4] {$10$};

        \draw[-latex,line width=0.5pt] (0,0.7) -- (c0.north);
	\draw[-latex,line width=0.5pt] (c0.east) --node [auto]{$c$} (c1.west);
	\draw[-latex,line width=0.5pt] (c1.east) --node [auto]{$b$} (c2.west);
    \draw[-latex,line width=0.5pt] (c2.east) --node [auto]{$b$} (c3.west);
    \draw[-latex,line width=0.5pt] (c3.east) --node [auto]{$f_1$} (c4.west);
	\draw[-latex,line width=0.5pt] (c0.west) --node [auto]{$\tau$} (c5.east);
    \draw[-latex,line width=0.5pt] (c4.east) --node [auto]{$c$} (c6.west);

    \draw[-latex,line width=0.5pt] (6.2,-.21)arc[start angle=-90, end angle=90, x radius=.32, y radius=.2];
	\node[black,below] at (6.6,0.2) {{$c$}};

	\end{tikzpicture}
    \caption{$\mathrm{supCNP}^2(L(G_1), \Psi_f^{-2}, \Sigma_{uc},P)$.}
    \label{fig:supG1}
\end{figure}

    \subsection{Active diagnosis}
In a similar way, according to Proposition \ref{prop:max}, we can enforce diagnosability by enforcing pre-normality of $L_f^{\geq N_o}$ w.r.t. $L(G)$ and $P$.
Given a DFA recognizing $L$, the sets of fault events $\Sigma_f$ and uncontrollable events $\Sigma_{uc}$, the number of observer states $N_o$, and natural projection $P$, the supremal diagnosable sublanguage of $G$ that is controllable and normal w.r.t. $G$ and projection $P$ is denoted by $\mathrm{supCND}(L, L_{f}^{\geq N_o}, \Sigma_{uc}, P)$. Then, inspired by Theorem \ref{Theorem:sup_prognosis_l}, this language can be computed by taking $M = L_{f}^{\geq N_o}$ in Step 2 of Algorithm~\ref{Alg: pro}. 

\begin{theorem}\label{Theorem:sup_diagnosable_l}
    \rm
Assume that $P(\Psi(\Sigma_f))$ is finite. Given a DFA $G$ recognizing $L=L(G)$, the supremal diagnosable sublanguage of $G$ that is controllable and normal w.r.t. $G$ and projection $P$ is $\mathrm{supCND}(L, L_{f}^{\geq N_o}, \Sigma_{uc},P)$ computed by Algorithm \ref{Alg: pro} with $M = L_{f}^{\geq N_o}$.
\end{theorem}
\begin{proof}
It follows from Theorem \ref{Theorem:sup_prognosis_l} and Proposition~\ref{prop:max}. 
\end{proof}

Note that if $P(\Psi((\Sigma_f))$ is not finite, due to Proposition~\ref{Prop: diaPsi_f}, then we can still use the sufficient condition, namely $L_{f}^{\geq N_o}$ is pre-normal w.r.t. $L$ to enforce diagnosability.
Although the supremal supervisors enforcing $k$-prognosability and diagnosability, given by $\mathrm{supCNP}^k(L,$ $\Psi_f^{-k},\Sigma_{uc}, P)$ and $\mathrm{supCND}(L,L_{f}^{\geq N_o},\Sigma_{uc}, P)$, always exist and can be computed using Algorithm \ref{Alg: pro}, their exponential complexity renders them impractical for large-scale modular systems. To address this limitation, we introduce a modular synthesis approach in the next section. 
\section{Active prognosis and diagnosis for modular DESs}\label{Section:modularenforcement}

In this section, we extend active prognosis/diagnosis from monolithic to modular DESs. Given a modular system $G=\parallel_{i=1}^l G_i$, and a component $G_i=(Q_i,\Sigma_i, \delta_i, q_{0,i})$ with $\Sigma_i=\Sigma_{i,o}\cup \Sigma_{i,uo}$ and $\Sigma_{i,f}\subseteq \Sigma_{i,uo}$, for $i=1,\dots ,l$, let the set of faults be $\Sigma_f=\bigcup_{i=1}^l \Sigma_{i,f}$.
We denote the faulty and non-faulty sublanguages of $G_i$, by $L_{i,f}$ and $L_{i,n}$ with $L_{i,f}\cup L_{i,n}=L(G_i)$, respectively. Then, the faulty and non-faulty sublanguages of the global plant $G$ are defined by 
\begin{align}\label{faulty_composition}
L_f = &\, L_{1,f}\parallel L_2\parallel L_3\ldots L_l \cup L_1\parallel L_{2,f}\parallel L_3\ldots L_l \cup \ldots \notag \\
      &\, \cup L_1\parallel L_2\parallel\ldots L_{l,f} \tag{3}
\end{align}
and $L_n=\parallel_{i=1}^l L_{i,n}$, respectively. 
Let $\Psi (\Sigma_{i,f}) = L_i\cap \Sigma_i^* \Sigma_{i,f}$ denote the set of strings that end with a fault event in a component $G_i$. The set of strings ending with a fault in the global plant $G=\parallel_{i=1}^l G_i$ is defined by $\Psi (\Sigma_{f}) = \Psi (\Sigma_{1,f})\parallel L_2\parallel L_3\ldots L_l\cup L_1\parallel \Psi (\Sigma_{2,f})\parallel L_3\ldots L_l\cup \ldots \cup L_1\parallel L_2\parallel\ldots \Psi (\Sigma_{l,f})$.

\subsection{Modular active prognosis}
 
Similar to the monolithic case, our approach is based on pre-normality of the prefix-closed languages $ L_{i,n}\setminus \Psi_{i,f}^{-k}$ w.r.t. $L_i$, where $\Psi_{i,f}^{-k}=P^{i^{-1}}_{i,o}P^i_{i,o}[ \Psi(\Sigma_{i,f}) / P^{i^{-1}}_{i,o}(\Sigma_{i,o}^{\leq k})]\Sigma_i^* \cap \overline{\Psi(\Sigma_{i,f})}$.  
It is known that pre-normality and prefix-closedness are preserved by the synchronous product, namely, if
for $i=1,\dots l$, $K_i$ are pre-normal (resp. prefix-closed) w.r.t. $L_i$ and $P^i_{i,o}$ then 
 $\parallel_{i=1}^l K_i$ is pre-normal (resp. prefix-closed) w.r.t. $\parallel_{i=1}^l L_i$ and $P$.
Now, we show that the pre-normality is also preserved for the composition of the type defined in Eq.~(\ref{faulty_composition}) when defining global faulty language based on local ones.

 
\begin{lemma} \label{Coro:||prenormal}
    Let $G = \|^l_{i=1} G_i$ be a modular plant with $L_i = L(G_i)$ over $\Sigma_i$, $i=1,\dots, l$ with $l \geq 2$, and let $L=L(G)$. For every $K_i$ that is pre-normal w.r.t. $L_i$ and $P^i_{i,o}$, it holds that  
 $K_1\|L_2 \| L_3 \dots L_l \cup   L_1 \|K_2 \| L_3 \dots L_l   \cup \dots \cup L_1\|L_2 \| L_3 \dots K_l $ is pre-normal w.r.t. $\parallel_{i=1}^l L_i$ and $P$.
\end{lemma}

\begin{proof}
 For simplicity, we prove the conclusion for two components, because the property can be extended to general $l \geq 2$. We have that $P_{1,0}^{1^{-1}}P_{1,o}^1(K_1)\cap L_1\subseteq K_1$ and $P_{2,0}^{2^{-1}}P_{2,o}^1(K_2)\cap L_2\subseteq K_2$. It holds that
 \begin{align*}
 &P^{-1}P[(K_1\parallel L_2) \cup (L_1\| K_2)] \cap (L_1\parallel L_2)\subseteq \\
 &[P^{-1}P(K_1|| L_2)\cap (L_1|| L_2)] \cup [P^{-1}P(L_1|| K_2)\cap (L_1|| L_2)]\\
 &\subseteq {P^{1^{-1}}_{1,o}}{P^1_{1,o}}(K_1)||{P^{2^{-1}}_{2,o}}{P^2_{2,o}}(L_2)\cap (L_1|| L_2)~\cup\\
 &~~~~{P^{1^{-1}}_{1,o}}{P^1_{1,o}}(L_1) || {P^{2^{-1}}_{2,o}}{P^2_{2,o}}(K_2)\cap (L_1|| L_2)\\
 &=[({P^{1^{-1}}_{1,o}}{P^1_{1,o}}(K_1)\cap L_1)|| L_2] \cup [L_1 || ({P^{2^{-1}}_{2,o}}{P^2_{2,o}}(K_2)\cap L_2)] \\
 &= (K_1\parallel L_2) \cup (L_1\parallel K_2).
 \end{align*}
 We conclude that $K_1\|L_2 \| L_3 \dots L_l \cup  \dots \cup L_1\|L_2 \| L_3 \dots K_l $ is pre-normal w.r.t. $\parallel_{i=1}^l L_i$ and $P$ if $K_i$ that is pre-normal w.r.t. $L_i$ and $P^i_{i,o}$ for $i=1,\dots, l$.
\end{proof}

Lemma \ref{Coro:||prenormal} implies that pre-normality is preserved under composition of faulty languages.
 We can construct local supervisors enforcing $k$-prognosability yielding  $k$-prognosable $L'_i \subseteq L_i$, i.e., $L'_{i,n}\setminus \Psi_{i,f}^{-k}$ is prefix-closed and pre-normal w.r.t. $L'_i$ for $i=1,2,\ldots, l$. We require that for the global plant $ L'_{n}\setminus \Psi_{f}^{-k}$ should be prefix-closed and pre-normal w.r.t. $L'$. 
 However, the global plant $G=\parallel_{i=1}^l G_i$ may not be $k$-prognosable w.r.t. fault $\Sigma_f=\bigcup_{i=1}^l \Sigma_{i,f}$ even if $G_i$ is $k$-prognosable w.r.t. fault $\Sigma_{i,f}$ for all $i=1,2,\ldots, l$. To this end, this section provides an approach to enforce the modular standard prognosability, i.e., $0$-prognosability. Before presenting the formal conclusion on modular active prognosis, we first introduce the following result, which simplifies the definition of the language $L_{n}\setminus \Psi_{f}^{-0}$.

\begin{lemma}\label{Coro:||-0}
Let $G = \|^l_{i=1} G_i$ be a modular plant with $L_i = L(G_i)$ over $\Sigma_i$, $i=1,\dots, l$ with $l \geq 2$, and let $L=L(G)$. For every $K_i\subseteq L_i$, we have
 $L \setminus \parallel_{i=1}^l K_i= (L_1\setminus K_1)\|L_2 \| L_3 \dots L_l \cup   L_1 \|(L_2\setminus K_2) \| L_3 \dots L_l   \cup \dots \cup L_1\|L_2 \| L_3 \dots (L_l\setminus K_l) $.
\end{lemma}

\begin{proof}
For simplicity, we prove the conclusion for $l=2$, since the property can be extended to general $l \geq 2$. 
 It holds 
 \begin{align*}
 &(L_1\| L_2)\setminus (K_1\| K_2)\\
 &=[(L_1\| L_2)\setminus P_{1}^{-1}(K_1)] \cup [(L_1\| L_2)\setminus P_{2}^{-1}(K_2)]  \\
 &=[(L_1\| L_2)\setminus (K_1\| \Sigma_2^*)] \cup [(L_1\| L_2)\setminus (\Sigma_1^*\| K_2)]\\
 &=[(L_1\| L_2)\setminus (K_1\| L_2)] \cup [(L_1\| L_2)\setminus (L_1\| K_2)]\\
 &=[(L_1\setminus K_1)\| L_2] \cup [L_1\|(L_2\setminus K_2)],
 \end{align*}
 which completes the proof.
\end{proof}

 Let $\Psi_{i,f}^{-0}=P_{i,o}^{i^{-1}}P_{i,o}^{i}(\Psi(\Sigma_{i,f}))\cap \overline{\Psi(\Sigma_{i,f})}$ be the language consisting of local strings that look like a string leading to the first fault of $G_i$. It holds $L_n\cap \Psi_f^{-0}=\|_{i=1}^l (L_{i,n}\cap \Psi_{i,f}^{-0})$.
 According to Lemma~\ref{Coro:||-0}, we have 
 $L_n\setminus \Psi_f^{-0}=L_n\setminus (L_n\cap \Psi_f^{-0})=(L_{1,n}\setminus \Psi_{1,f}^{-0})\parallel L_2\parallel L_3\ldots L_l\cup L_1\parallel (L_{2,n}\setminus \Psi_{2,f}^{-0}) \parallel L_3\ldots L_l\cup L_1\parallel L_2\ldots (L_{l,n}\setminus \Psi_{l,f}^{-0}) $, which is the sublanguage of the global plant that we require to be prefix-closed and pre-normal w.r.t. $\parallel_{i=1}^l L_i$.
In the following, we show that $0$-prognosability is preserved under the synchronous product. 

\begin{proposition}\label{Coro:||0pro}
   Let $G = \|^l_{i=1} G_i$ be a modular plant with $L_i = L(G_i)$ over $\Sigma_i$, $i=1,\dots, l$ with $l \geq 2$, and let $L=L(G)$. The global plant $G=\parallel_{i=1}^l G_i$ is prognosable w.r.t. fault $\Sigma_f$ and projection $P$ if for every $i=1,2,\ldots, l$, $G_i$ is prognosable w.r.t. fault $\Sigma_{i,f}$ and projection $P_{i,o}^i$. 
\end{proposition}
 
\begin{proof}
    According to Proposition \ref{Pro:prodef} and Lemma~\ref{Coro:||-0}, we need to show that the language $L_n\setminus \Psi_f^{-0}$ is prefix-closed and pre-normal w.r.t. $L=\parallel_{i=1}^l L_i$ and $P$ if languages $L_{i,n}\setminus \Psi_{i,f}^{-0}$ are prefix-closed and pre-normal w.r.t. $L_i$ and $P_{i,o}^i$ for $i=1,2,\ldots, l$. By  Lemma~\ref{Coro:||prenormal}, 
    $(L_{1,n}\setminus \Psi_{1,f}^{-0})\parallel L_2\parallel L_3\ldots L_l\cup L_1\parallel (L_{2,n}\setminus \Psi_{2,f}^{-0}) \parallel L_3\ldots L_l\cup L_1\parallel L_2\ldots (L_{l,n}\setminus \Psi_{l,f}^{-0}) $ is
     pre-normal w.r.t. $\parallel_{i=1}^l L_i$ and $P$.
   Further, since prefix-closedness is also preserved under the synchronous product and the other type of the product used for faulty languages (cf. Eq.~(\ref{faulty_composition})), we have that language $L_n\setminus \Psi_f^{-0}$ is prefix-closed.
   We conclude that $L_n\setminus \Psi_f^{-0}$ is prefix-closed and pre-normal w.r.t. $L$ and $P$, which is equivalent to that $G$ is prognosable w.r.t. $\Sigma_f$ and $P$ by Proposition~\ref{Pro:prodef} and Corollary~\ref{Corollary:equiv}. 
\end{proof}

Proposition~\ref{Coro:||0pro} implies that if event component $G_i$ is prognosable w.r.t. $\Sigma_{i,f}$ and $P_{i,o}^i$, then their parallel composition $G=\parallel_{i=1}^l G_i$ is prognosable w.r.t. $\Sigma_f$ and $P$.
To achieve modular prognosability enforcement, we first enforce prognosability for every component by computing the supremal controllable, normal, and $0$-prognosable sublanguage $\mathcal{S}_i=\supCNP^0(L_i,\Psi_{i,f}^{-0},\Sigma_{i,uc}, P_{i,o}^i)$. Then, we show that the global prognosability can be achieved through the parallel composition of local supervisors. 
 
\begin{theorem}\label{Theorem:modularenforce_pro}
Let $G = \|^l_{i=1} G_i$ be a modular plant with $L=L(G)$, $L_i = L(G_i)$ over $\Sigma_i$, for $i=1,\dots, l$, where $l \geq 2$. If languages $\supCN(\mathcal{S}_i, L_i, \Sigma_{i,uc},P_{i,o}^i)$ are nonconflicting for all $i=1,\dots, l$, then, $\|^l_{i=1} \mathcal{S}_i$ is controllable, normal and prognosable w.r.t. $L$, $\Sigma_f$ and $P$.
\end{theorem}
\begin{proof}
Controllability and normality of 
$\|^l_{i=1} \mathcal{S}_i$ follows from Theorem 3 in \cite{komenda2023modular}, and prognosability
follows from Proposition~\ref{Coro:||0pro}.
\end{proof}

\begin{example} \label{Exa:modularpro}
Consider again two DFAs $G_1=(Q_1, \Sigma_1, \delta_1,$ $ q_{0,1})$ and $G_2=(Q_2, \Sigma_2, \delta_2, q_{0,2})$ depicted in Figs. \ref{Fig:G_1}(a) and \ref{G_2}(a). For simplicity, assume that $E_{c}=E_o$, i.e., $E_{uc}=E_{uo}$.
Since $G_1$ is prognosable by Example~\ref{Exa:prodef}, we change the arc $6\xrightarrow{a}7$ by $6\xrightarrow{b}7$. In this way, by observing $abb$ with $abbf_1\in \Psi(\Sigma_{1,f})$, one cannot infer whether fault $f_1$ will occur. Thus, the revised $G_1$ is not prognosable. Further, $G_2$ is not prognosable since one cannot infer whether fault $f_2$ will occur by observing $a$ with $af_2\in\Psi(\Sigma_{2,f})$.  

Now we show how to enforce the global plant $G_1\parallel G_2$ to be prognosable by the proposed approach. 
For $G_1$, by Algorithm~\ref{Alg: pro}, a DFA $H_1=(X_1,\Sigma,\Delta_1,x_0)$ with $H_1\sqsubset G_1$ and $L(H_1)=L(G_1)$ is constructed as depicted in Fig.~\ref{Fig:G_r1}(a). Its observer $Obs(H_1)$ is portrayed in Fig. \ref{Fig:G_r1}(b). Let the marked states be the states reached by firing strings in $L_{1,f}\cup \Psi_{1,f}^{-0}$, i.e., $X_{m}=\{3,4,10,3',4',10'\}$, which are shown in red in Fig. \ref{Fig:G_r1}(a). According to $Obs(H_1)$, the observer states that contain both marked and unmarked states are $\{3,4,8\}$ and $\{10,8'\}$. It holds that $X^{N_p}=X_1\setminus\{3,4,8,10,8'\}$. At the first iteration, we compute the condition $X_2^C=X_2^N=\{0,1,2,5,6,7,9,2',3',4',10'\}$. After removing all unreached states and the corresponding arcs, we have $X_3= X_2^C=X_2^N$ and $L(H_3)$ as shown in Fig.~\ref{Fig:G_r1}(c). At the second iteration, we derive $X_3^C=X_3^N=X_3$. Since the specification conditions remain unchanged, we conclude that $\mathrm{supCNP}^0(L_1, \Psi_{1,f}^{-0}, \Sigma_{1,uc},P_{1,o}^{1})=L(H_3)$ depicted in Fig.~\ref{Fig:G_r1}(c). In the same way, for $G_2$, the supremal controllable, normal, and prognosable sublanguage $\mathrm{supCNP}^0(L_2, \Psi_{2,f}^{-0}, \Sigma_{2,uc},P_{2,o}^{2})$ is depicted in Fig.~\ref{Fig:G_r1}(d). Finally, by Theorem~\ref{Theorem:modularenforce_pro}, $\mathrm{supCNP}^0(L_1, \Psi_{1,f}^{-0}, \Sigma_{i,uc},P_{1,o}^{1})\parallel \mathrm{supCNP}^0(L_2, \Psi_{2,f}^{-0}, \Sigma_{2,uc},P_{2,o}^{2})$ is controllable, normal, and prognosable w.r.t. $L$, $\Sigma_f$, and $P$.
    \hfill$\blacksquare$
\end{example}

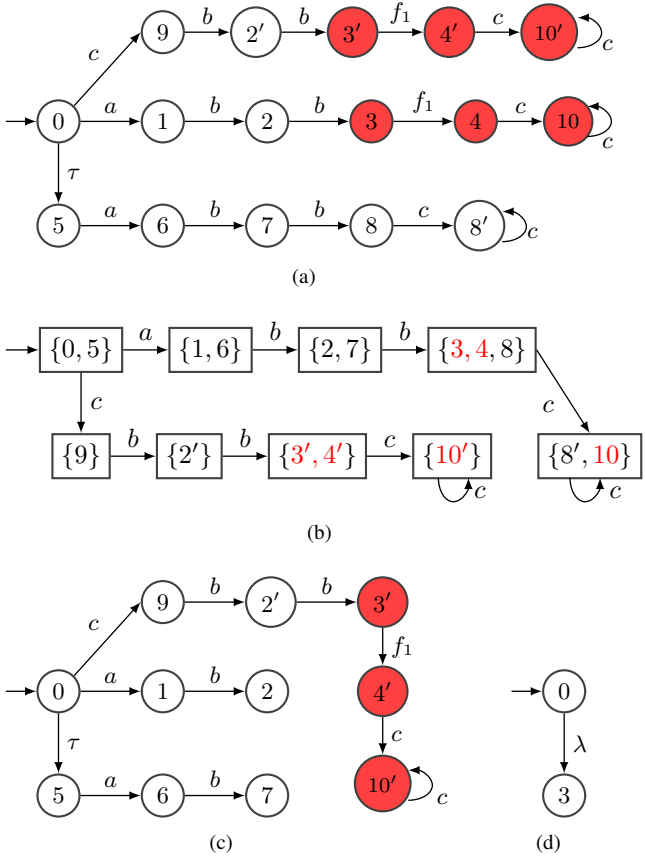
\begin{figure}[H]
\subfigure[]
{
    \centering
	\begin{tikzpicture}{every node}=[font= \fontsize{9pt}{10pt},scale=1]
	\node (c0) [circle,thick,draw=black!75] {$0$};
	\node (c1) [circle,thick,draw=black!75,right=0.8cm of c0] {$1$};
	\node (c2) [circle,thick,draw=black!75,right=0.8cm of c1] {$2$};
    \node (c3) [circle,thick,draw=black!75,fill=red!75,right=0.8cm of c2] {$3$};
	\node (c4) [circle,thick,draw=black!75,fill=red!75,right=0.8cm of c3] {$4$};
	\node (c5) [circle,thick,draw=black!75,below=0.8cm of c0] {$5$};
	\node (c6) [circle,thick,draw=black!75,right=0.8cm of c5] {$6$};
    \node (c7) [circle,thick,draw=black!75,right=0.8cm of c6] {$7$};
	\node (c8) [circle,thick,draw=black!75,right=0.8cm of c7] {$8$};
    \node (c9) [circle,thick,draw=black!75,above=0.6cm of c1] {$9$};
    \node (c10) [circle,thick,draw=black!75,right=0.6cm of c9] {$2'$};
    \node (c11) [circle,thick,draw=black!75,fill=red!75,right=0.6cm of c10] {$3'$};
	\node (c12) [circle,thick,draw=black!75,fill=red!75,right=0.6cm of c11] {$4'$};
    \node (c13) [circle,thick,draw=black!75,right=0.8cm of c8] {$8'$};
    \node (c14) [circle,thick,draw=black!75,fill=red!75,font=\footnotesize,right=0.6cm of c4] {$10$};
    \node (c15) [circle,thick,draw=black!75,fill=red!75,font=\footnotesize,right=0.6cm of c12] {$10'$};

        \draw[-latex,line width=0.5pt] (-0.7,0) -- (c0.west);
	\draw[-latex,line width=0.5pt] (c0.east) --node [auto]{$a$} (c1.west);
    \draw[-latex,line width=0.5pt] (c0.north east) --node [auto]{$c$} (c9.west);
    \draw[-latex,line width=0.5pt] (c9.east) --node [auto]{$b$} (c10.west);
    \draw[-latex,line width=0.5pt] (c10.east) --node [auto]{$b$} (c11.west);
    \draw[-latex,line width=0.5pt] (c11.east) --node [auto]{$f_1$} (c12.west);
	\draw[-latex,line width=0.5pt] (c1.east) --node [auto]{$b$} (c2.west);
    \draw[-latex,line width=0.5pt] (c2.east) --node [auto]{$b$} (c3.west);
    \draw[-latex,line width=0.5pt] (c3.east) --node [auto]{$f_1$} (c4.west);
	\draw[-latex,line width=0.5pt] (c8.east) --node [auto]{$c$} (c13.west);
	\draw[-latex,line width=0.5pt] (c0.south) --node [auto]{$\tau$} (c5.north);
    \draw[-latex,line width=0.5pt] (c6.east) --node [auto]{$b$} (c7.west);
    \draw[-latex,line width=0.5pt] (c7.east) --node [auto]{$b$} (c8.west);
	\draw[-latex,line width=0.5pt] (c5.east) --node [auto]{$a$} (c6.west);

    \draw[-latex,line width=0.5pt] (c4.east) --node [auto]{$c$} (c14.west);
    \draw[-latex,line width=0.5pt] (c12.east) --node [auto]{$c$} (c15.west);
	\draw[-latex,line width=0.5pt] (5.9,-1.6)arc[start angle=-90, end angle=90, x radius=.32, y radius=.2];
	\node[black,below] at (6.3,-1.3) {{$c$}};

    \draw[-latex,line width=0.5pt] (7.03,-0.2)arc[start angle=-90, end angle=90, x radius=.32, y radius=.2];
	\node[black,below] at (7.3,-0.1) {{$c$}};

    \draw[-latex,line width=0.5pt] (6.9,0.98)arc[start angle=-90, end angle=90, x radius=.32, y radius=.2];
	\node[black,below] at (7.3,1.2) {{$c$}};
	\end{tikzpicture}}
    \subfigure[]{
    \centering
\begin{tikzpicture}
    \node (c0) [draw,thick,draw=black!75] {$\{0,5\}$};
    \node (c1) [draw,thick,draw=black!75, right=0.6cm of c0] {$\{1,6\}$};
    \node (c2) [draw,thick,draw=black!75, right=0.6cm of c1] {$\{2,7\}$};
    \node (c3) [draw,thick,draw=black!75, right=0.6cm of c2] {$\{\textcolor{red}{3,4},8\}$};
    \node (c8) [draw,thick,draw=black!75, below=0.8cm of c0] {$\{9\}$};
    \node (c5) [draw,thick,draw=black!75, right=0.6cm of c8] {$\{2'\}$};
    \node (c6) [draw,thick,draw=black!75, right=0.6cm of c5] {$\{\textcolor{red}{3',4'}\}$};
    \node (c9) [draw,thick,draw=black!75, right=0.6cm of c6] {$\{\textcolor{red}{10'}\}$};
    \node (c7) [draw,thick,draw=black!75, right=0.6cm of c9] {$\{8',\textcolor{red}{10}\}$};

    \draw[-latex,line width=0.5pt] (-1,0) -- (c0.west);
    \draw[-latex,line width=0.5pt] (c0.east) --node [auto]{$a$} (c1.west);
    \draw[-latex,line width=0.5pt] (c1.east) --node [auto]{$b$} (c2.west);
    \draw[-latex,line width=0.5pt] (c2.east) --node [auto]{$b$} (c3.west);
    \draw[-latex,line width=0.5pt] (c5.east) --node [auto]{$b$} (c6.west);
    \draw[-latex,line width=0.5pt] (c8.east) --node [auto]{$b$} (c5.west);
    \draw[-latex,line width=0.5pt] (c0.south) --node [auto]{$c$} (c8.north);
    \draw[-latex,line width=0.5pt] (c3.east) --node [auto,swap]{$c$} (c7.north);
    \draw[-latex,line width=0.5pt] (c6.east) --node [auto]{$c$} (c9.west);

    
    \draw[-latex,line width=0.5pt] (4.75,-1.7) arc[start angle=-180, end angle=0, x radius=.2, y radius=.32];
    \node[black,below] at (5.3,-1.7) {{$c$}};

    \draw[-latex,line width=0.5pt] (6.5,-1.7) arc[start angle=-180, end angle=0, x radius=.2, y radius=.32];
    \node[black,below] at (7.1,-1.7) {{$c$}};

	\end{tikzpicture}
    }
    \subfigure[]
{
    \centering
	\begin{tikzpicture}{every node}=[font= \fontsize{9pt}{10pt},scale=1]
	\node (c0) [circle,thick,draw=black!75] {$0$};
	\node (c1) [circle,thick,draw=black!75,right=0.8cm of c0] {$1$};
	\node (c2) [circle,thick,draw=black!75,right=0.8cm of c1] {$2$};
	\node (c5) [circle,thick,draw=black!75,below=0.8cm of c0] {$5$};
	\node (c6) [circle,thick,draw=black!75,right=0.8cm of c5] {$6$};
    \node (c7) [circle,thick,draw=black!75,right=0.8cm of c6] {$7$};
    \node (c9) [circle,thick,draw=black!75,above=0.6cm of c1] {$9$};
    \node (c10) [circle,thick,draw=black!75,right=0.8cm of c9] {$2'$};
    \node (c11) [circle,thick,draw=black!75,fill=red!75,right=0.8cm of c10] {$3'$};
	\node (c12) [circle,thick,draw=black!75,fill=red!75,below=0.5cm of c11] {$4'$};
    \node (c13) [circle,thick,draw=black!75,fill=red!75,font=\footnotesize,below=0.5cm of c12] {$10'$};

        \draw[-latex,line width=0.5pt] (-0.7,0) -- (c0.west);
	\draw[-latex,line width=0.5pt] (c0.east) --node [auto]{$a$} (c1.west);
    \draw[-latex,line width=0.5pt] (c0.north east) --node [auto]{$c$} (c9.west);
    \draw[-latex,line width=0.5pt] (c9.east) --node [auto]{$b$} (c10.west);
    \draw[-latex,line width=0.5pt] (c10.east) --node [auto]{$b$} (c11.west);
    \draw[-latex,line width=0.5pt] (c11.south) --node [auto]{$f_1$} (c12.north);
    \draw[-latex,line width=0.5pt] (c12.south) --node [auto]{$c$} (c13.north);
	\draw[-latex,line width=0.5pt] (c1.east) --node [auto]{$b$} (c2.west);
	\draw[-latex,line width=0.5pt] (c0.south) --node [auto]{$\tau$} (c5.north);
    \draw[-latex,line width=0.5pt] (c6.east) --node [auto]{$b$} (c7.west);
	\draw[-latex,line width=0.5pt] (c5.east) --node [auto]{$a$} (c6.west);

    \draw[-latex,line width=0.5pt] (4.65,-1.45)arc[start angle=-90, end angle=90, x radius=.32, y radius=.2];
	\node[black,below] at (5.1,-1.25) {{$c$}};

	\end{tikzpicture}}
    ~~~
    \subfigure[]
	{
	\centering
	\begin{tikzpicture}{every node}=[font= \fontsize{9pt}{10pt},scale=1]
	\node (c0) [circle,thick,draw=black!75] {$0$};
	\node (c3) [circle,thick,draw=black!75,below=0.8cm of c0] {$3$};
    \draw[-latex,line width=0.5pt] (-0.7,0) -- (c0.west);
	\draw[-latex,line width=0.5pt] (c0.south) --node [auto]{$\lambda$} (c3.north);
	\end{tikzpicture}	}
    \caption{(a) A DFA $H_1$ with $H_1 \sqsubset G_1$, $L(H_1)=L(G_1)$ and $\delta_1(6,b)=7$, (b) the observer $Obs(H_1)$, (c) $\mathrm{supCNP}^0(L_1, \Psi_{1,f}^{-0}, \Sigma_{1,uc},P_{1,o}^{1})$, and (d) $\mathrm{supCNP}^0(L_2, \Psi_{2,f}^{-0}, \Sigma_{2,uc},P_{2,o}^{2})$.}
    \label{Fig:G_r1}
\end{figure}

\subsection{Modular active diagnosis}
In this subsection, we emphasize that due to the characterization of diagnosability as pre-normality of a suffix of the global faulty language (under the additional assumption that $P(\Psi((\Sigma_f))$ is finite), one can always (even without the above finiteness assumption as the sufficient condition is enough for enforcement) use the above modular prognosability enforcement to modular diagnosability enforcement.

Given a modular system $G=\parallel_{i=1}^l G_i$ and a component $G_i$, let $L_i^{\geq N_{i,o}}=L_{i}\| \Sigma_{i,o}^{\geq N_{i,o}}$ and $L_{i,f}^{\geq N_{i,o}}=L_{i,f}\| \Sigma_{i,o}^{\geq N_{i,o}}$ be the plant sublanguage and faulty sublanguage containing $N_{i,o}$ observations of $G_i$, respectively, where $N_{i,o}$ is the number of local observer states of $G_i$ and $\Sigma_{i,o}^{\geq N_{i,o}}=\{t\in \Sigma_{i,o}^*\mid |P_{i,o}^i(t)|\geq N_{i,o}\}$ . We have the following result for a suffix of the global faulty language $L_f$.

\begin{lemma}\label{Coro:||-dia}
Let $G = \|^l_{i=1} G_i$ be a modular plant with $L=L(G)$, $L_i = L(G_i)$ over $\Sigma_i$, for $i=1,\dots, l$, where $l \geq 2$. There exists a natural number $\bar{N}\in\mathbb{N}$ such that $L_f^{\geq \bar{N}} = L_{1,f}^{\geq N_{1,o}}\parallel L_2^{\geq N_{2,o}}\parallel L_3^{\geq N_{3,o}}\ldots L_l^{\geq N_{l,o}} \cup L_1^{\geq N_{1,o}}\parallel L_{2,f}^{\geq N_{2,o}}\parallel L_3^{\geq N_{3,o}}\ldots L_l^{\geq N_{l,o}} \cup \ldots \notag \cup L_1^{\geq N_{1,o}}\parallel L_2^{\geq N_{2,o}}\parallel\ldots L_{l,f}^{\geq N_{l,o}}$.
\end{lemma}

\begin{proof}
    It is obvious that there exists a natural number $\bar{N}\in\mathbb{N}$ with $\max\{N_{1,o},N_{2,o},\ldots, N_{l,o}\}\leq \bar{N} \leq \sum_{i=1}^l N_{i,o}$ such that $\Sigma_o^{\geq \bar{N}}=\Sigma_{1,0}^{\geq N_{1,o}}\| \Sigma_{2,0}^{\geq N_{2,o}}\| \ldots \| \Sigma_{l,0}^{\geq N_{l,o}}$. By Eq.~(\ref{faulty_composition}) and distributivity of the synchronous product with language unions, we have 
    \begin{align*}
& L_f^{\geq \bar{N}}
 = L_f \parallel \Sigma_o^{\geq \bar{N}} \\
 &= \big( L_{1,f}\parallel L_2 \parallel \cdots \parallel L_l
     \ \cup\  L_1 \parallel L_{2,f} \parallel \cdots \parallel L_l
     \ \cup\  \cdots \\
 &\quad \cup\  L_1 \parallel L_2 \parallel \cdots \parallel L_{l,f} \big)
     \parallel \big( \Sigma_{1,o}^{\geq N_{1,o}} \parallel \cdots \parallel
     \Sigma_{l,o}^{\geq N_{l,o}} \big) \\
 &= L_{1,f}^{\geq N_{1,o}}\parallel L_2^{\geq N_{2,o}}\parallel \cdots \parallel L_l^{\geq N_{l,o}}
 \cup\ L_1^{\geq N_{1,o}}\parallel L_{2,f}^{\geq N_{2,o}}\parallel\\
 &\cdots \parallel L_l^{\geq N_{l,o}}
 \cup\ \cdots \cup\ L_1^{\geq N_{1,o}}\parallel L_2^{\geq N_{2,o}}\parallel \cdots \parallel L_{l,f}^{\geq N_{l,o}},
\end{align*}
which completes the proof.
\end{proof}

It is emphasized that if $P(\Psi((\Sigma_f))$ is not finite, due to Proposition \ref{Prop: ldia}, then we can still use the sufficient condition, namely to enforce  that $L_f^{\geq \bar{N}}$ is pre-normal w.r.t. $L$ to guarantee diagnosability, i.e., by computing the supremal controllable, normal and diagnosable sublanguage $L'=\supCND(L,L_f^{\geq \bar{N}},\Sigma_{uc},P)$ of $L$. Now we show that diagnosability is also preserved under the composition of suffixes of faulty languages as in Lemma \ref{Coro:||-dia}.

\begin{proposition}\label{Coro:||dia}
    Let $G = \|^l_{i=1} G_i$ be a modular plant with $L_i = L(G_i)$ over $\Sigma_i$, $i=1,\dots, l$ with $l \geq 2$, and let $L=L(G)$. The global plant $G=\parallel_{i=1}^l G_i$ is diagnosable w.r.t. fault $\Sigma_f$ and projection $P$ if for every $i=1,2,\ldots, l$, $G_i$ is diagnosable w.r.t. fault $\Sigma_{i,f}$ and projection $P_{i,o}^i$. 
\end{proposition}

\begin{proof}
By  Lemma \ref{Coro:||-dia}, there exists a non-negative integer $\bar{N}\in\mathbb{N}$ with $\max\{N_{1,o},N_{2,o},\ldots, N_{l,o}\}\leq \bar{N} \leq \sum_{i=1}^l N_{i,o}$ such that language $L_f^{\geq \bar{N}}$
has a similar (union of synchronous product) form as in Lemma~\ref{Coro:||prenormal}, just with $L_i$ replaced by $L_i^{\geq N_{i,o}}$.  
From the definition of pre-normality and from $L_i^{\geq N_{i,o}}\subseteq L_i$ for all $i\in \{1,2,\ldots, l\}$, it holds that language $L_{i,f}^{\geq N_{i,o}}$ is pre-normal w.r.t. $L_i^{\geq N_{i,o}}$ and $P_{i,o}^i$ if it is pre-normal w.r.t. $L_i$ and $P_{i,o}^i$. 
Due to diagnosability of local plants and Proposition \ref{Prop: ldia}, the languages $L_{i,f}^{\geq N_{i,o}}$ are pre-normal w.r.t. $L_i$ and $P_{i,o}^i$ for $i=1,2,\ldots, l$. From Lemmas~\ref{Coro:||prenormal} and~\ref{Coro:||-dia},  $L_f^{\geq \bar{N}}$  
is pre-normal w.r.t. $L=\parallel_{i=1}^l L_i$ and $P$. According to  Proposition \ref{Prop: ldia}, $G$ is diagnosable w.r.t. $\Sigma_f$ and $P$. 
\end{proof}

To achieve modular diagnosability enforcement, we enforce diagnosability for every component by computing the supremal controllable, normal, and diagnosable sublanguage w.r.t. $L_i$, i.e., $\mathcal{S}_i=\supCND(L_i,L_{i,f}^{\geq N_{i,o}},\Sigma_{i,uc}, P_{i,o}^i)$. 

\begin{theorem}\label{Theorem:modularenforce_dia}
Let $G = \|^l_{i=1} G_i$ be a modular plant with $L=L(G)$, $L_i = L(G_i)$ over $\Sigma_i$, for $i=1,\dots, l$, where $l \geq 2$. If the languages $\supCN(\mathcal{S}_i, L_i, \Sigma_{i,uc},P_{i,o}^i)$ are nonconflicting for all $i=1,\dots, l$, then $\parallel_{i=1}^l \mathcal{S}_i$ is globally controllable, normal, and diagnosable w.r.t. $L$, $\Sigma_f$ and $P$.
\end{theorem}
\begin{proof}
    It follows from Proposition \ref{Coro:||dia} and Theorem~\ref{Theorem:modularenforce_pro}.
\end{proof}

\section{Conclusion}\label{Sec: Conclusion}
We provide a novel characterization of $k$-prognosability (resp. diagnosability) in terms of pre-normality of a superlanguage (resp. suffix) of the faulty language. It is shown that $k$-prognosability implies $0$-prognosabilty and $0$-prognosability is equivalent to the standard prognosability. 
Moreover, we prove the existence of the supremal $k$-prognosable/diagnosable and normal sublanguage, and develop an algorithm to compute the supremal controllable, normal, and $k$-prognosable/diagnosable sublanguage for a monolithic plant.
This algorithm for active $k$-prognosis/diagnosis can be extended to modular DESs and does not suffer from the weaknesses of online active diagnosis approaches, where the computations need to be carried out faster than the system's evolution. 

Our next goal is to provide conditions under which modular (off-line) enforcement of prognosability/diagnosability is not more restrictive than enforcement of prognosability/diagnosability for the monolithic plant.


\bibliographystyle{IEEEtran}
\bibliography{references}

\end{document}